\documentclass[11pt]{article}
\usepackage[left=1in, right=1in, top=1in, bottom=1in]{geometry}
\usepackage{CJK}
\usepackage{latexsym,bm}
\usepackage{xypic}
\usepackage{amsmath}
\usepackage{wasysym}
\usepackage{indentfirst}
\usepackage{amssymb}
\usepackage{dsfont}
\usepackage{amsthm}
\begin{document}
\newcommand{\fr}[2]{\frac{\;#1\;}{\;#2\;}}
\newtheorem{theorem}{Theorem}[section]
\newtheorem{lemma}{Lemma}[section]
\newtheorem{proposition}{Proposition}[section]
\newtheorem{corollary}{Corollary}[section]
\newtheorem{conjecture}{Conjecture}[section]
\newtheorem{remark}{Remark}[section]
\newtheorem{definition}{Definition}[section]
\newtheorem{example}{Example}[section]
\newtheorem{notation}{Notation}[section]
\numberwithin{equation}{section}
\newcommand{\Aut}{\mathrm{Aut}\,}
\newcommand{\CSupp}{\mathrm{CSupp}\,}
\newcommand{\Supp}{\mathrm{Supp}\,}
\newcommand{\rank}{\mathrm{rank}\,}
\newcommand{\col}{\mathrm{col}\,}
\newcommand{\len}{\mathrm{len}\,}
\newcommand{\leftlen}{\mathrm{leftlen}\,}
\newcommand{\rightlen}{\mathrm{rightlen}\,}
\newcommand{\length}{\mathrm{length}\,}
\newcommand{\bin}{\mathrm{bin}\,}
\newcommand{\wt}{\mathrm{wt}\,}
\newcommand{\Wt}{\mathrm{Wt}\,}
\newcommand{\diff}{\mathrm{diff}\,}
\newcommand{\lcm}{\mathrm{lcm}\,}
\newcommand{\GL}{\mathrm{GL}\,}
\newcommand{\SJ}{\mathrm{SJ}\,}
\newcommand{\LG}{\mathrm{LG}\,}
\newcommand{\bij}{\mathrm{bij}\,}
\newcommand{\dom}{\mathrm{dom}\,}
\newcommand{\fun}{\mathrm{fun}\,}
\newcommand{\SUPP}{\mathrm{SUPP}\,}
\newcommand{\supp}{\mathrm{supp}\,}
\newcommand{\End}{\mathrm{End}\,}
\newcommand{\Hom}{\mathrm{Hom}\,}
\newcommand{\ran}{\mathrm{ran}\,}
\newcommand{\row}{\mathrm{row}\,}
\newcommand{\Mat}{\mathrm{Mat}\,}
\newcommand{\rk}{\mathrm{rk}\,}
\newcommand{\rs}{\mathrm{rs}\,}
\newcommand{\piv}{\mathrm{piv}\,}
\newcommand{\perm}{\mathrm{perm}\,}
\newcommand{\rsupp}{\mathrm{rsupp}\,}
\newcommand{\inv}{\mathrm{inv}\,}
\newcommand{\orb}{\mathrm{orb}\,}
\newcommand{\id}{\mathrm{id}\,}
\newcommand{\soc}{\mathrm{soc}\,}
\newcommand{\unit}{\mathrm{unit}\,}
\newcommand{\word}{\mathrm{word}\,}

\renewcommand{\thefootnote}{\fnsymbol{footnote}}

\title{$r$-Minimal Poset Codes}
\author{Jianhua Zheng$^{1,2}$\,\,\,\,\,\,Yang Xu$^1$ \,\,\,\,\,\,Yidan Chen$^2$ \,\,\,\,\,\, Haibin Kan$^{1,3}$\,\,\,\,\,\,Guangyue Han$^4$}

\maketitle

\renewcommand{\thefootnote}{\fnsymbol{footnote}}

\footnotetext{\hspace*{-6mm} \begin{tabular}{@{}r@{}p{16cm}@{}}
&This work has been supported by the National and Local Joint Laboratory of Cyberspace Security Technology (No. NLCS-202512-003).\\
$^1$ & Shanghai Key Laboratory of Intelligent Information Processing, College of Computer Science and Artificial Intelligence, Fudan University,
Shanghai 200433, China.\\
&Shanghai Engineering Research Center of Blockchain, Shanghai 200433, China. {E-mail:xuyyang@fudan.edu.cn, hbkan@fudan.edu.cn}\\
$^2$ & National and Local Joint Laboratory of Cyberspace Security Technology.\\
$^3$ &Yiwu Research Institute of Fudan University, Yiwu City, Zhejiang 322000, China. \\
$^4$ & Department of Mathematics, Faculty of Science, The University of Hong Kong, Pokfulam Road, Hong Kong, China. {E-mail:ghan@hku.hk} \\
\end{tabular}}

\vskip 3mm

{\hspace*{-6mm}{\bf Abstract}---In this paper, we propose and study $r$-minimal codes with respect to $\mathbf{P}$-support, where $\mathbf{P}=(\Omega,\preccurlyeq_{\mathbf{P}})$ is a poset defined on the coordinate set of the ambient space $\mathbf{H}$. $r$-Minimal $\mathbf{P}$-codes are natural extensions of Hamming metric minimal codes that have been extensively studied in the literature. We characterize $r$-minimal $\mathbf{P}$-codes in terms of the notion so called cutting $r$-blocking maps, which generalizes the well-known equivalence between minimal Hamming metric codes and cutting blocking sets. We also give a necessary and sufficient condition for $r$-minimality in terms of $(\mathbf{P},\omega)$-weight defined on $\mathbf{H}$, where $\omega:\Omega\longrightarrow\mathbb{R}^{+}$ is an arbitrary weight function. This leads to a generalization of the well-known Ashikhmin-Barg criterion for Hamming metric minimal codes. We then prove two existence results for $r$-minimal $\mathbf{P}$-codes, both for general $\mathbf{P}$ and for the special case that $\mathbf{P}$ is a disjoint union of chains. When $\mathbf{P}$ is hierarchical, we characterize $r$-minimal $\mathbf{P}$-codes in terms of $r$-minimal Hamming metric codes. Finally, we characterize cutting $r$-blocking sets induced by hierarchical posets with two levels, which further enables us to answer a question raised in Hyun, Kim, Wu and Yue \cite{28}.

\section{Introduction}
\setlength{\parindent}{2em}
Throughout the paper, let $\mathbb{F}$ be a finite field with $|\mathbb{F}|=q$, $\Omega$ be a nonempty finite set with $|\Omega|=n$, and let
$$\mathbf{H}\triangleq\mathbb{F}^{\Omega}.$$
Any $\mathbb{F}$-subspace $C\leqslant_{\mathbb{F}}\mathbf{H}$ is referred to as a \textit{code}. For any $\alpha\in\mathbf{H}$, the \textit{Hamming support of $\alpha$}, denoted by $\supp(\alpha)$, is defined as
\begin{equation}\supp(\alpha)=\{i\in\Omega\mid\alpha_i\neq0\}.\end{equation}
With respect to Hamming support, a nonzero codeword $\beta\in C$ is said to be minimal in $C$ if there is no nonzero codeword $\alpha\in C$ with $\supp(\alpha)\subsetneqq\supp(\beta)$; moreover, the code $C$ is said to be minimal if all of its nonzero codewords are minimal. In early 1990s, Massey showed in \cite{36} that the minimal codewords in the dual code completely specify the access structure of the secret sharing scheme based on a code. Minimal codes have since been a topic of great interest in coding theory due to their applications in secret sharing, secure two-party computation and constructing decoding algorithms (see, e.g., \cite{5,13,14,17,37,57}). A well-known sufficient condition for minimality has been given in Ashikhmin and Barg \cite{5}, namely, a nonzero code $C$ is minimal provided that
\begin{equation}\frac{v_{min}}{v_{max}}>1-q^{-1},\end{equation}
where $v_{min}$ and $v_{max}$ denote the minimal and maximal Hamming weights of all the nonzero codewords of $C$, respectively. Many classes of minimal codes have since been constructed with the help of the Ashikhmin-Barg condition (1.2) (see, e.g., \cite{13,17,18,20,37,38,46,57,58}). The first infinite family of binary minimal codes that violates the Ashikhmin-Barg condition has been given in Chang and Hyun \cite{14} by using simplicial complexes. Later in \cite{19,25}, Heng, Ding and Zhou have given a necessary and sufficient condition for minimality, namely, a code $C$ is minimal if and only if
\begin{equation}\sum_{c\in\mathbb{F}}|\supp(\alpha+c\beta)|\neq(q-1)|\supp(\beta)|\end{equation}
for any two linearly independent codewords $\alpha,\beta\in C$. With the help of this result, many classes of binary and ternary minimal codes, satisfying or violating the Ashikhmin-Barg condition, have been explored by using Boolean and vectorial functions (see, e.g., \cite{19,25,28,31,32,53}).

Another characterization of minimal codes arises from the combinatorial notion of cutting blocking sets, which have been introduced in Bonini and Borello \cite{10} for constructing minimal codes violating the Ashikhmin-Barg condition. It has been proven independently in Alfarano, Borello and Neri \cite{1} and Tang, Qiu, Liao, Zhou \cite{47} that a code $C$ is minimal if and only if all the columns of a given generator matrix of $C$ form a cutting blocking set. Cutting blocking sets are very powerful tools for studying minimal codes. For example, using cutting blocking sets, Alfarano, Borello, Neri and Ravagnani have given recursive upper bounds for minimal length of minimal codes and explicit constructions of minimal codes in \cite{2}; Alon, Bishnoi, Das and Neri have given explicit constructions of asymptotically good families of minimal codes in \cite{4} (c.f. \cite{9}); H\'{e}ger and Nagy have established upper bounds for minimal length of minimal codes which are linear in both the dimension and the field size in \cite{23}; later in \cite{8}, Bishnoi, D' haeseleer, Gijswijt and Potukuchi have further generalized the result in \cite{23} to more general settings (c.f. \cite{3}).

In this paper, we propose and study $r$-minimal codes with respect to $\mathbf{P}$-support, where $\mathbf{P}=(\Omega,\preccurlyeq_{\mathbf{P}})$ is a poset defined on the coordinate set $\Omega$, and $r\in\mathbb{N}$ stands for $r$-dimensional $\mathbb{F}$-subspaces of a given code. When $\mathbf{P}$ is an anti-chain and $r=1$, $r$-minimal $\mathbf{P}$-codes are exactly minimal Hamming metric codes (see Section 2.1 for more details). In coding theory, weight and metrics induced by poset structures have been a research topic of great interest since poset metric being introduced in Brualdi, Graves and Lawrence \cite{12} (c.f. \cite{40,41}). Various important coding-theoretic properties such as MacWilliams identity (\cite{21,22,29,34,35,44,48,54}), MacWilliams extension property (\cite{6,22,24,34,35,55}), perfect codes and MDS codes (\cite{26,27,43}), group of isometries (\cite{22,30,35,42,43,55}) and generalized weights (\cite{39,51}) have been explored with respect to various weights and metrics induced by poset structures. To the best of our knowledge, other than having been mentioned in \cite{56} briefly, minimal codes have not been explored in the context of poset support yet.

The main contributions of this paper can be summarized as follows.

In Section 3, we study $r$-minimal $\mathbf{P}$-codes where $\mathbf{P}$ is an arbitrary poset. More precisely, in Section 3.1, we give necessary and sufficient conditions for a code to be $r$-minimal in terms of $(r+1)$-dimensional subcodes (Proposition 3.1, Theorem 3.1), and derive several corollaries including a singleton-type bound (Corollary 3.1). In Section 3.2, we characterize $r$-minimal $\mathbf{P}$-codes in terms of cutting $r$-blocking maps with respect to $\mathbf{P}$ (Theorem 3.2), which, with $\mathbf{P}$ set to be an anti-chain, boils down to the equivalence between $r$-minimal codes with respect to Hamming support and cutting $r$-blocking sets (Corollary 3.5), which generalizes the equivalence between minimal codes and cutting blocking sets when $r=1$. In Section 3.3, we give necessary and sufficient conditions for a code to be $r$-minimal in terms of $(\mathbf{P},\omega)$-weight, where $\omega:\Omega\longrightarrow\mathbb{R}^{+}$ is an arbitrary weight function (Proposition 3.5, Theorem 3.3), which lead to generalizations of (1.2) and (1.3). In Section 3.4, we derive a general existence results for $r$-minimal $\mathbf{P}$-codes, which is then applied to the special case that $\mathbf{P}$ is a disjoint union of chains (Theorem 3.4, Corollary 3.8).

In Section 4, we focus on the special case that $\mathbf{P}$ is hierarchical and characterize $r$-minimal $\mathbf{P}$-codes in terms of $r$-minimal codes with respect to Hamming support (Proposition 4.1, Theorem 4.1).

In Section 5, we study several sets induced by a hierarchical poset with two levels, and give necessary and sufficient conditions for these sets to be cutting $r$-blocking sets (Theorems 5.1--5.4). Via the equivalence between cutting $r$-blocking sets and $r$-minimal codes, our results provide necessary and sufficient conditions for codes constructed from these sets to be $r$-minimal with respect to Hamming support, which further enables us to answer a question raised in Hyun, Kim, Wu and Yue \cite{28}.

\section{Preliminaries}
\setlength{\parindent}{2em}
We begin with a few more notations that will be used throughout the rest of the paper. For any $A\subseteq\mathbf{H}$, following \cite{51}, the \textit{Hamming support of $A$}, denoted by $\chi(A)$, is defined as
\begin{equation}\chi(A)=\{i\in\Omega\mid\exists~\alpha\in A~s.t.~\alpha_i\neq0\}.\end{equation}
For any $I\subseteq\Omega$, define
\begin{equation}\delta(I)=\{\beta\in\mathbf{H}\mid\supp(\beta)\subseteq I\}.\end{equation}
For an $\mathbb{F}$-vector space $X$ and any $A\subseteq X$, let $\langle A\rangle_{\mathbb{F}}$ denote the $\mathbb{F}$-subspace of $X$ generated by $A$. For $k\in\mathbb{Z}^{+}$, let $\mathbb{F}^{k}$ and $\mathbb{F}^{[k]}$ denote the sets of all the row vectors and column vectors over $\mathbb{F}$ of length $k$, respectively, and let $\Mat_{k,\Omega}(\mathbb{F})$ denote the set of all the matrices over $\mathbb{F}$ whose rows are indexed by $\{1,\dots,k\}$ and columns are indexed by $\Omega$; moreover, for any $G\in\Mat_{k,\Omega}(\mathbb{F})$, let $\row(G)\subseteq\mathbf{H}$ and $\col(G)\subseteq\mathbb{F}^{[k]}$ denote the sets of all the rows and columns of $G$, respectively. For a code $C\leqslant_{\mathbb{F}}\mathbf{H}$ with $\dim_{\mathbb{F}}(C)=k\geqslant1$, a \textit{generator matrix of $C$} is a matrix $G\in\Mat_{k,\Omega}(\mathbb{F})$ satisfying that
\begin{equation}C=\langle\row(G)\rangle_{\mathbb{F}}=\{\gamma G\mid \gamma\in\mathbb{F}^{k}\}.\end{equation}
For $k\in\mathbb{Z}^{+}$ and $A\subseteq \mathbb{F}^{k}$, let
\begin{equation}A^{\bot}\triangleq\{\beta\in \mathbb{F}^{[k]}\mid \mbox{$\sum_{i=1}^{k}\alpha_i\beta_i=0$ for all $\alpha\in A$}\}.\end{equation}
Finally, for $k\in\mathbb{Z}^{+}$ and $G\in\Mat_{k,\Omega}(\mathbb{F})$, the \textit{column map of $G$} is the map $f:\Omega\longrightarrow\mathbb{F}^{[k]}$ defined as
\begin{equation}\forall~i\in\Omega:\text{$f(i)$ is equal to the $i$-th column of $G$}.\end{equation}

\subsection{$r$-Minimal poset codes and cutting $r$-blocking maps}
\setlength{\parindent}{2em}
Throughout this subsection, let $\mathbf{P}=(\Omega,\preccurlyeq_{\mathbf{P}})$ be a poset. A subset $B\subseteq \Omega$ is referred to as an \textit{ideal} of $\mathbf{P}$ if for any $v\in B$ and $u\in \Omega$, $u\preccurlyeq_{\mathbf{P}}v$ implies that $u\in B$. The set of all the ideals of $\mathbf{P}$ is denoted by $\mathcal{I}(\mathbf{P})$. For $B\subseteq \Omega$, we let $\max_{\mathbf{P}}(B)$ (\textit{resp.}, $\min_{\mathbf{P}}(B)$) denote the set of all the maximal (\textit{resp.}, minimal) elements of $B$, and let
$$\langle B\rangle_{\mathbf{P}}\triangleq\{u\in \Omega\mid \exists~v\in B~s.t.~u\preccurlyeq_{\mathbf{P}}v\};$$
moreover, $B$ is said to be a \textit{chain} in $\mathbf{P}$ if for any $u,v\in B$, either $u\preccurlyeq_{\mathbf{P}}v$ or $v\preccurlyeq_{\mathbf{P}}u$ holds, and $B$ is said to be an \textit{anti-chain} in $\mathbf{P}$ if for any $u,v\in B$, $u\preccurlyeq_{\mathbf{P}}v$ implies $u=v$. For any $u\in \Omega$, we let $\len_{\mathbf{P}}(u)$ denote the largest cardinality of a chain in $\mathbf{P}$ containing $u$ as its greatest element. The set of all the order automorphisms of $\mathbf{P}$ will be denoted by $\Aut(\mathbf{P})$. The \textit{dual poset} of $\mathbf{P}$ is defined as $\mathbf{\overline{P}}=(\Omega,\preccurlyeq_{\mathbf{\overline{P}}})$, where
$$\text{$u\preccurlyeq_{\mathbf{\overline{P}}} v\Longleftrightarrow v\preccurlyeq_{\mathbf{P}}u$ for all $(u,v)\in \Omega\times \Omega$}.$$
In order to define $r$-minimal codes with respect to $\mathbf{P}$, we consider the \textit{$\mathbf{P}$-support} $\langle\chi(A)\rangle_{\mathbf{P}}$ defined for any $A\subseteq\mathbf{H}$.

\setlength{\parindent}{0em}
\begin{definition}
{\bf{(1)}}\,\,Let $C\leqslant_{\mathbb{F}}\mathbf{H}$. For $D\leqslant_{\mathbb{F}}C$, we say that $D$ is $\mathbf{P}$-minimal in $C$ if for any $Q\leqslant_{\mathbb{F}}C$ such that $\dim_{\mathbb{F}}(Q)=\dim_{\mathbb{F}}(D)$, $\langle\chi(Q)\rangle_{\mathbf{P}}\subseteq\langle\chi(D)\rangle_{\mathbf{P}}$, it holds that $\langle\chi(Q)\rangle_{\mathbf{P}}=\langle\chi(D)\rangle_{\mathbf{P}}$. Moreover, for $r\in\mathbb{N}$, we say that $C$ is $r$-minimal with respect to $\mathbf{P}$ if every $r$-dimensional $\mathbb{F}$-subspace of $C$ is $\mathbf{P}$-minimal in $C$.

{\bf{(2)}}\,\,Let $C\leqslant_{\mathbb{F}}\mathbf{H}$. For $D\leqslant_{\mathbb{F}}C$, we say that $D$ is Hamming minimal in $C$ if for any $Q\leqslant_{\mathbb{F}}C$ such that $\dim_{\mathbb{F}}(Q)=\dim_{\mathbb{F}}(D)$, $\chi(Q)\subseteq\chi(D)$, it holds that $\chi(Q)=\chi(D)$. Moreover, for $r\in\mathbb{N}$, we say that $C$ is $r$-minimal with respect to Hamming support if every $r$-dimensional $\mathbb{F}$-subspace of $C$ is Hamming minimal in $C$.
\end{definition}

\begin{remark}
$r$-Minimal codes with respect to Hamming support have also been defined in \cite{9} recently (see [9, Section 2.2]). When $\mathbf{P}$ is an anti-chain, $\mathbf{P}$-support becomes Hamming support, and hence (2) is in fact a special case of (1); moreover, when $r=1$, (2) recovers the definitions for minimal codewords and minimal codes with respect to Hamming support (see, e.g., [5, Section II.A] and [2, Definition 1.3]).
\end{remark}

\setlength{\parindent}{2em}
Next, we recall the definition of cutting $r$-blocking sets.

\setlength{\parindent}{0em}
\begin{definition}([10, Definitions 3.2 and 3.4], [2, Definitions 1.9 and 1.10]) Let $k\in\mathbb{N}$, and let $X$ be a $k$-dimensional $\mathbb{F}$-vector space. For $S\subseteq X$ and $r\in\mathbb{N}$, $S$ is referred to as a cutting $r$-blocking set of $X$ if $\langle S\cap V\rangle_{\mathbb{F}}=V$ for all $V\leqslant_{\mathbb{F}}X$ with $\dim_{\mathbb{F}}(V)=k-r$.
\end{definition}

\setlength{\parindent}{2em}
Now, as an extension of Definition 2.2, we define cutting $r$-blocking maps with respect to $\mathbf{P}$.

\setlength{\parindent}{0em}
\begin{definition}
Let $X$ be a $k$-dimensional $\mathbb{F}$-vector space. For $f:\Omega\longrightarrow X$ and $r\in\mathbb{N}$, we say that $f$ is a cutting $r$-blocking map of $X$ with respect to $\mathbf{P}$ if for any $V\leqslant_{\mathbb{F}}X$ with $\dim_{\mathbb{F}}(V)=k-r$, it holds that
\begin{equation}\langle\{f(i)\mid i\in\Omega,(\forall~j\in\Omega:i\preccurlyeq_{\mathbf{P}}j\Longrightarrow f(j)\in V)\}\rangle_{\mathbb{F}}=V.\end{equation}
\end{definition}

\begin{remark}
If $\mathbf{P}$ is an anti-chain, then $f$ is a cutting $r$-blocking map of $X$ with respect to $\mathbf{P}$ if and only if $\ran(f)=\{f(i)\mid i\in\Omega\}$ is a cutting $r$-blocking set of $X$. It has been shown in \cite{1,47} that a code $C$ is $1$-minimal with respect to Hamming support if and only if $\col(G)$ is  a cutting $1$-blocking set. In Section 3.2, we will generalize this equivalence to $r$-minimal $\mathbf{P}$-codes and cutting $r$-blocking maps with respect to $\mathbf{P}$.
\end{remark}

\setlength{\parindent}{2em}
We illustrate Definitions 2.2 and 2.3 with the following example.

\begin{example}
Set $q=2$, $\Omega=\{1,2,3,4,5,6,7\}$, and let $\mathbf{P}$ be the poset defined as $1\preccurlyeq_{\mathbf{P}}2$, $3\preccurlyeq_{\mathbf{P}}4$, $5\preccurlyeq_{\mathbf{P}}6$. Then, for the matrix
\begin{displaymath}G=\left(\begin{array}{ccccccc}
                   1&1&1&0&0&1&0\\
                   1&0&1&1&1&1&0\\
                   1&0&1&0&1&0&1\\
                   \end{array}\right),
\end{displaymath}
it is straightforward to verify that the column map $f$ of $G$ is a cutting $1$-blocking map of $\mathbb{F}_{2}^{[3]}$ with respect to $\mathbf{P}$, and $\col(G)$ is a cutting $1$-blocking set of $\mathbb{F}_{2}^{[3]}$.
\end{example}

\setlength{\parindent}{2em}
The following proposition will be used in Section 3.

\begin{proposition}
Let $X$ be a $k$-dimensional $\mathbb{F}$-vector space, $r\in\{0,1,\dots,k-1\}$, and let $f:\Omega\longrightarrow X$ be a cutting $r$-blocking map of $X$ with respect to $\mathbf{P}$. Then, for any $s\in\{0,1,\dots,r\}$, $f$ is a cutting $s$-blocking map of $X$ with respect to $\mathbf{P}$.
\end{proposition}

\begin{proof}
Let $s\in\{0,1,\dots,r\}$, and let $V\leqslant_{\mathbb{F}}\mathbb{F}^{[k]}$ with $\dim_{\mathbb{F}}(V)=k-s$. Since $1\leqslant k-r\leqslant k-s$, we have $V=\bigcup_{(U\leqslant_{\mathbb{F}}V,\dim_{\mathbb{F}}(U)=k-r)}U$; moreover, for any $U\leqslant_{\mathbb{F}}V$ with $\dim_{\mathbb{F}}(U)=k-r$, since (2.6) holds for $U$, we have $U\subseteq\langle\{f(i)\mid i\in\Omega,(\forall~j\in\Omega:i\preccurlyeq_{\mathbf{P}}j\Longrightarrow f(j)\in V)\}\rangle_{\mathbb{F}}$. Therefore we conclude that (2.6) holds for $V$, as desired.
\end{proof}

We end this subsection by recalling two classes of posets which have been extensively studied for poset codes.

\setlength{\parindent}{0em}
\begin{definition}
{\bf{(1)}}\,\,$\mathbf{P}$ is said to be hierarchical if for any $u,v\in \Omega$ such that $\len_{\mathbf{P}}(u)+1\leqslant\len_{\mathbf{P}}(v)$, it holds that $u\preccurlyeq_{\mathbf{P}}v$.

{\bf{(2)}}\,\,For $s,m\in\mathbb{Z}^{+}$, $\mathbf{P}$ is referred to as an $(s,m)$ Niederreiter-Rosenbloom-Tsfasman (NRT) poset if $\mathbf{P}$ is isomorphic to $(\{1,\dots,s\}\times\{1,\dots,m\},\preccurlyeq_{_{NRT}})$, where $(a,b)\preccurlyeq_{_{NRT}}(c,d)\Longleftrightarrow(a=c,~b\leqslant d)$.
\end{definition}

\begin{remark}
{\bf{(1)}}\,\,All the hierarchical posets over $\Omega$ are in one-to-one correspondence with all the ordered partitions of $\Omega$. More precisely, for $s\in\mathbb{Z}^{+}$ and a tuple of non-empty disjoint sets $(A_1,\dots,A_s)$ such that $\bigcup_{i=1}^{s}=\Omega$, the poset $\mathbf{Q}=(\Omega,\preccurlyeq_{\mathbf{Q}})$ defined as $u\preccurlyeq_{\mathbf{Q}}v$ if and only if $u=v$ or $u\in A_i$, $v\in A_j$ for some $i<j\in\{1,\dots,s\}$ is hierarchical; conversely, any hierarchical poset over $\Omega$ can be obtained in this way. Hierarchical posets have many coding-theoretic characterizations. For example, among others, it has been proven in \cite{6,29,34} that $\mathbf{P}$ is hierarchical if and only if $\mathbf{P}$ admits a MacWilliams identity, if and only if $\mathbf{P}$ satisfies the MacWilliams extension property, if and only if $\mathbf{P}$ induces an association scheme, if and only if the group of $\mathbf{P}$-weight isometries acts transitively on codewords with the same $\mathbf{P}$-weight. We refer the reader to \cite{22,26,28,35,42,44,53,54,55} and references therein for more results established for hierarchical posets.

{\bf{(2)}}\,\,An $(s,m)$ NRT poset can be regarded as the disjoint union of $s$ chains of length $m$; in particular, when $s=1$ (\textit{resp}., $m=1$), an $(s,m)$ NRT poset boils down to a chain (\textit{resp}., an anti-chian). Metric spaces induced by NRT posets were implicitly introduced in \cite{40,41} for studying combinatorial problems for vector spaces over finite fields, and it was later noted in \cite{12} that these spaces are in fact a special class of poset metric spaces. We refer the reader to \cite{6,21,24,30,43,45,48} and references therein for relevant results for NRT poset codes.
\end{remark}

\subsection{Weighted poset metric and constant weight codes}

\setlength{\parindent}{2em}
Weighted poset metric is a general class of metric that has been introduced by Hyun, Kim and Park in \cite{27}, where the authors have classified all the weighted posets that admit the extended Hamming code to be a $2$-perfect code. In Section 3, we will characterize $r$-minimal poset codes in terms of weighted poset metric, and show that $r$-constant weight codes with respect to weighted poset metric are in fact a special class of $r$-minimal poset codes.

Throughout the rest of this subsection, we fix a poset $\mathbf{P}=(\Omega,\preccurlyeq_{\mathbf{P}})$ and a weight function $\omega:\Omega\longrightarrow\mathbb{R}^{+}$, and consider the \textit{$\omega$-weighted poset} $(\mathbf{P},\omega)$ (see \cite{27}). For any $\beta\in\mathbf{H}$, the $(\mathbf{P},\omega)$-weight of $\beta$ is defined as
\begin{equation}\wt_{(\mathbf{P},\omega)}(\beta)\triangleq\sum_{i\in\langle\supp(\beta)\rangle_{\mathbf{P}}}\omega(i).\end{equation}
More generally, for any $A\subseteq\mathbf{H}$, we define the $(\mathbf{P},\omega)$-weight of $A$ as
\begin{equation}\Wt_{(\mathbf{P},\omega)}(A)\triangleq\sum_{i\in\langle\chi(A)\rangle_{\mathbf{P}}}\omega(i).\end{equation}
The metric $d_{(\mathbf{P},\omega)}:\mathbf{H}\times \mathbf{H}\longrightarrow \mathbb{R}$ defined as
\begin{equation}d_{(\mathbf{P},\omega)}(\alpha,\beta)=\wt_{(\mathbf{P},\omega)}(\beta-\alpha)\end{equation}
is referred to as \textit{weighted poset metric}. We note that as the weight function $\omega$ takes values on each coordinate position, weighted poset metric can be useful to model some specific channels in which the error probability depends on a codeword position (i.e., the distribution of errors is nonuniform), and can also be useful to perform bitwise or messagewise unequal error protection (see, e.g., the abstract of \cite{7} and [22, Section I, Paragraph VI]). We refer the reader to \cite{22,27,35,54,55} for more relevant results for weighted poset metric.

In the following two examples, we give two special cases of $(\mathbf{P},\omega)$-weight.

\begin{example}($\mathbf{P}$-weight)
If $\omega(i)=1$ for all $i\in\Omega$, then (2.7) and (2.8) boil down to the $\mathbf{P}$-weight defined as
\begin{equation}\wt_{\mathbf{P}}(\beta)\triangleq\wt_{(\mathbf{P},\omega)}(\beta)=|\langle\supp(\beta)\rangle_{\mathbf{P}}|,\end{equation}
\begin{equation}\Wt_{\mathbf{P}}(A)\triangleq\Wt_{(\mathbf{P},\omega)}(A)=|\langle\chi(A)\rangle_{\mathbf{P}}|,\end{equation}
respectively; moreover, (2.9) recovers the notion of $\mathbf{P}$-metric (see \cite{12,39}). In particular, if $\mathbf{P}$ is an NRT poset, then (2.9) becomes the NRT metric (see \cite{21,45}).
\end{example}

\begin{example}($\omega$-weight)
If $\mathbf{P}$ is an anti-chain, then (2.7) and (2.8) boil down to the $\omega$-weight defined as
\begin{equation}\wt_{\omega}(\beta)\triangleq\wt_{(\mathbf{P},\omega)}(\beta)=\sum_{i\in\supp(\beta)}\omega(i),\end{equation} \begin{equation}\Wt_{\omega}(A)\triangleq\Wt_{(\mathbf{P},\omega)}(A)=\sum_{i\in\chi(A)}\omega(i),\end{equation}
respectively; moreover, (2.9) recovers the notion of weighted Hamming metric (see \cite{7}).
\end{example}

\setlength{\parindent}{2em}
Now we define $r$-constant weight code with respect to $(\mathbf{P},\omega)$.

\setlength{\parindent}{0em}
\begin{definition}
For $C\leqslant_{\mathbb{F}}\mathbf{H}$ and $r\in\mathbb{N}$, we say that $C$ is $r$-constant weight with respect to $(\mathbf{P},\omega)$ if for any $D,Q\leqslant_{\mathbb{F}}C$ with $\dim_{\mathbb{F}}(Q)=\dim_{\mathbb{F}}(D)=r$, it holds that $\Wt_{(\mathbf{P},\omega)}(D)=\Wt_{(\mathbf{P},\omega)}(Q)$. In particular, we say that $C$ is $r$-constant weight with respect to $\omega$ (\textit{resp}., $\mathbf{P}$) if for any $D,Q\leqslant_{\mathbb{F}}C$ with $\dim_{\mathbb{F}}(Q)=\dim_{\mathbb{F}}(D)=r$, it holds that $\Wt_{\omega}(D)=\Wt_{\omega}(Q)$ (\textit{resp}., $\Wt_{\mathbf{P}}(D)=\Wt_{\mathbf{P}}(Q)$).
\end{definition}

\begin{remark}
Constant weight Hamming metric codes are first classified by Bonisoli in \cite{11} where it has been proven that every constant weight Hamming metric code is a sequence of dual Hamming codes. Different proofs of this result have been given in \cite{33} by using the notion of value function, in \cite{49} by using characters, and in \cite{50} by using the MacWilliams extension property (c.f. \cite{52}).
\end{remark}

\setlength{\parindent}{2em}
Now we introduce linear maps that preserve $(\mathbf{P},\omega)$-weight.

\setlength{\parindent}{0em}
\begin{definition}
An $\mathbb{F}$-automorphism $\varphi\in\End_{\mathbb{F}}(\mathbf{H})$ is referred to as a \textit{$(\mathbf{P},\omega)$-weight isometry} of $\mathbf{H}$ if $\wt_{(\mathbf{P},\omega)}(\varphi(\alpha))=\wt_{(\mathbf{P},\omega)}(\alpha)$ for all $\alpha\in \mathbf{H}$, and is referred to as a \textit{$\mathbf{P}$-support isometry} of $\mathbf{H}$ if $\langle\supp(\varphi(\alpha))\rangle_{\mathbf{P}}=\langle\supp(\alpha)\rangle_{\mathbf{P}}$ for all $\alpha\in \mathbf{H}$.
\end{definition}

\setlength{\parindent}{2em}
We collect some basic properties of $(\mathbf{P},\omega)$-weight isometries in the following lemma.

\setlength{\parindent}{0em}
\begin{lemma}([35, Theorem 5], [55, Theorem III.1]) Let $\varphi\in\End_{\mathbb{F}}(\mathbf{H})$ be a $(\mathbf{P},\omega)$-weight isometry. Then, there uniquely exists $\lambda\in\Aut(\mathbf{P})$ such that $\omega(i)=\omega(\lambda(i))$ for all $i\in\Omega$ and $\langle\supp(\varphi(\alpha))\rangle_{\mathbf{P}}=\lambda[\langle\supp(\alpha)\rangle_{\mathbf{P}}]$ for all $\alpha\in\mathbf{H}$. Moreover, for any $A\subseteq\mathbf{H}$, it holds true that $\langle\chi(\varphi[A])\rangle_{\mathbf{P}}=\lambda[\langle\chi(A)\rangle_{\mathbf{P}}]$ and $\Wt_{(\mathbf{P},\omega)}(\varphi[A])=\Wt_{(\mathbf{P},\omega)}(A)$.
\end{lemma}

\section{Characterizations of $r$-Minimal poset codes}
\setlength{\parindent}{2em}
Throughout this section, we fix a poset $\mathbf{P}=(\Omega,\preccurlyeq_{\mathbf{P}})$.

\subsection{Some basic properties}
We begin with the following characterization of $\mathbf{P}$-minimal subcodes.

\setlength{\parindent}{0em}
\begin{proposition}
For $C\leqslant_{\mathbb{F}}\mathbf{H}$ and $D\leqslant_{\mathbb{F}}C$, the following five statements are equivalent to each other:

{\bf{(1)}}\,\,$D$ is $\mathbf{P}$-minimal in $C$;

{\bf{(2)}}\,\,$C\cap\delta(\langle\chi(D)\rangle_{\mathbf{P}})=D$;

{\bf{(3)}}\,\,$\dim_{\mathbb{F}}(C)-\dim_{\mathbb{F}}(D)=\dim_{\mathbb{F}}(\delta(\langle\chi(D)\rangle_{\mathbf{P}})+C)-\Wt_{\mathbf{P}}(D)$;

{\bf{(4)}}\,\,For any $A\subseteq C$ with $\chi(A)\subseteq\langle\chi(D)\rangle_{\mathbf{P}}$, it holds that $A\subseteq D$;

{\bf{(5)}}\,\,For any $B\leqslant_{\mathbb{F}}C$ such that $\dim_{\mathbb{F}}(B)=\dim_{\mathbb{F}}(D)$, $\langle\chi(B)\rangle_{\mathbf{P}}\subseteq\langle\chi(D)\rangle_{\mathbf{P}}$, it holds that $B=D$.
\end{proposition}

\begin{proof}
$(1)\Longrightarrow(2)$\,\,Suppose by way of contradiction that $C\cap\delta(\langle\chi(D)\rangle_{\mathbf{P}})\neq D$. Then, from $D\subseteq C\cap\delta(\langle\chi(D)\rangle_{\mathbf{P}})$, we can choose $W\leqslant_{\mathbb{F}}C\cap\delta(\langle\chi(D)\rangle_{\mathbf{P}})$ such that $D\subseteq W$ and $\dim_{\mathbb{F}}(W)=\dim_{\mathbb{F}}(D)+1$. Choose $i\in\max_{\mathbf{P}}(\chi(W))$. Then, one can check that $B\triangleq\{\alpha\in W\mid\alpha_i=0\}$ satisfies that  $B\leqslant_{\mathbb{F}}W$, $\dim_{\mathbb{F}}(B)=\dim_{\mathbb{F}}(W)-1$ and $\langle\chi(B)\rangle_{\mathbf{P}}\subseteq\langle\chi(W)\rangle_{\mathbf{P}}-\{i\}\subsetneqq\langle\chi(W)\rangle_{\mathbf{P}}$. From $W\subseteq \delta(\langle\chi(D)\rangle_{\mathbf{P}})$, we deduce that $\langle\chi(W)\rangle_{\mathbf{P}}\subseteq\langle\chi(D)\rangle_{\mathbf{P}}$, which implies that $\langle\chi(B)\rangle_{\mathbf{P}}\subsetneqq\langle\chi(D)\rangle_{\mathbf{P}}$. Noticing that $\dim_{\mathbb{F}}(B)=\dim_{\mathbb{F}}(D)$, it follows from Definition 2.1 that $D$ is not $\mathbf{P}$-minimal in $C$, a contradiction, as desired.

$(2)\Longleftrightarrow(3)$\,\,We infer that (3) holds true if and only if $\dim_{\mathbb{F}}(D)=\dim_{\mathbb{F}}(C\cap\delta(\langle\chi(D)\rangle_{\mathbf{P}}))$, which, together with $D\subseteq C\cap\delta(\langle\chi(D)\rangle_{\mathbf{P}})$, immediately implies $(2)\Longleftrightarrow(3)$, as desired.

$(2)\Longrightarrow(4)$\,\,Let $A\subseteq C$ with $\chi(A)\subseteq\langle\chi(D)\rangle_{\mathbf{P}}$. Then, we have $A\subseteq\delta(\langle\chi(D)\rangle_{\mathbf{P}})$, which, together with (2), implies that $A\subseteq C\cap\delta(\langle\chi(D)\rangle_{\mathbf{P}})=D$, as desired.

$(4)\Longrightarrow(5)$ and $(5)\Longrightarrow(1)$\,\,These two parts are trivial.
\end{proof}

\setlength{\parindent}{2em}
The following theorem is the main result of this subsection.

\setlength{\parindent}{0em}
\begin{theorem}
For $C\leqslant_{\mathbb{F}}\mathbf{H}$ and $r\in\mathbb{N}$, the following four statements are equivalent to each other:

{\bf{(1)}}\,\,$C$ is $r$-minimal with respect to $\mathbf{P}$;

{\bf{(2)}}\,\,For any $D,B\leqslant_{\mathbb{F}}C$ such that $\dim_{\mathbb{F}}(D)=\dim_{\mathbb{F}}(B)=r$, $\langle\chi(B)\rangle_{\mathbf{P}}\subseteq\langle\chi(D)\rangle_{\mathbf{P}}$, it holds that $D=B$;

{\bf{(3)}}\,\,Every $(r+1)$-dimensional $\mathbb{F}$-subspace of $C$ is $r$-minimal with respect to $\mathbf{P}$;

{\bf{(4)}}\,\,For any $W\leqslant_{\mathbb{F}}C$ with $\dim_{\mathbb{F}}(W)=r+1$ and any $D\leqslant_{\mathbb{F}}W$ with $\dim_{\mathbb{F}}(D)=r$, it holds that $\langle\chi(D)\rangle_{\mathbf{P}}\subsetneqq\langle\chi(W)\rangle_{\mathbf{P}}$.
\end{theorem}

\begin{proof}
$(1)\Longleftrightarrow(2)$\,\,This follows from ``$(1)\Longleftrightarrow(5)$'' of Proposition 3.1.

$(1)\Longrightarrow(3)$\,\,This immediately follows from Definition 2.1.

$(3)\Longrightarrow(4)$\,\,Since $W$ is $r$-minimal with respect to $\mathbf{P}$, $D$ is $\mathbf{P}$-minimal in $W$. Therefore by Proposition 3.1, we have $\chi(W)\nsubseteq\langle\chi(D)\rangle_{\mathbf{P}}$, which implies that $\langle\chi(D)\rangle_{\mathbf{P}}\subsetneqq\langle\chi(W)\rangle_{\mathbf{P}}$, as desired.

$(4)\Longrightarrow(1)$\,\,By Proposition 3.1, it suffices to show that $C\cap\delta(\langle\chi(D)\rangle_{\mathbf{P}})=D$ for any $D\leqslant_{\mathbb{F}}C$ with $\dim_{\mathbb{F}}(D)=r$. Indeed, if $C\cap\delta(\langle\chi(D)\rangle_{\mathbf{P}})\neq D$, then we can choose $W\leqslant_{\mathbb{F}}C\cap\delta(\langle\chi(D)\rangle_{\mathbf{P}})$ such that $D\subseteq W$ and $\dim_{\mathbb{F}}(W)=r+1$; moreover, by (4), we have $\langle\chi(D)\rangle_{\mathbf{P}}\subsetneqq\langle\chi(W)\rangle_{\mathbf{P}}$, which implies that $\chi(W)\nsubseteq\langle\chi(D)\rangle_{\mathbf{P}}$, and hence $W\nsubseteq\delta(\langle\chi(D)\rangle_{\mathbf{P}})$, a contradiction, as desired.
\end{proof}

\setlength{\parindent}{2em}
Now we derive some corollaries of Proposition 3.1 and Theorem 3.1, and we begin with the following Singleton type bound.

\setlength{\parindent}{0em}
\begin{corollary}
{\bf{(1)}}\,\,Let $C\leqslant_{\mathbb{F}}\mathbf{H}$, and let $D\leqslant_{\mathbb{F}}C$ such that $D$ is $\mathbf{P}$-minimal in $C$. Then, for any $I\subseteq\Omega$ with $\langle\chi(D)\rangle_{\mathbf{P}}\cup\chi(C)\subseteq I$, it holds that
$$\dim_{\mathbb{F}}(C)-\dim_{\mathbb{F}}(D)\leqslant|I|-\Wt_{\mathbf{P}}(D),$$
where the equality holds if and only if $\delta(\langle\chi(D)\rangle_{\mathbf{P}})+C=\delta(I)$. In particular, it holds that
$$\dim_{\mathbb{F}}(C)-\dim_{\mathbb{F}}(D)\leqslant\Wt_{\mathbf{P}}(C)-\Wt_{\mathbf{P}}(D),$$
where the equality holds if and only if $\delta(\langle\chi(D)\rangle_{\mathbf{P}})+C=\delta(\langle\chi(C)\rangle_{\mathbf{P}})$.

{\bf{(2)}}\,\,Let $C\leqslant_{\mathbb{F}}\mathbf{H}$ with $\dim_{\mathbb{F}}(C)=k$, and let $r\in\{0,1,\dots,k\}$ such that $C$ is $r$-minimal with respect to $\mathbf{P}$. Then, for any $D\leqslant_{\mathbb{F}}C$ with $\dim_{\mathbb{F}}(D)=r$, it holds that $$\Wt_{\mathbf{P}}(D)\leqslant\Wt_{\mathbf{P}}(C)-k+r.$$
\end{corollary}

\begin{proof}
We only prove (1) from which (2) immediately follows. By Proposition 3.1, we have
$$\dim_{\mathbb{F}}(C)-\dim_{\mathbb{F}}(D)=\dim_{\mathbb{F}}(\delta(\langle\chi(D)\rangle_{\mathbf{P}})+C)-\Wt_{\mathbf{P}}(D).$$
Hence for any $I\subseteq\Omega$ with $\langle\chi(D)\rangle_{\mathbf{P}}\cup\chi(C)\subseteq I$, the desired result follows from $\delta(\langle\chi(D)\rangle_{\mathbf{P}})+C\leqslant_{\mathbb{F}}\delta(I)$ and $\dim_{\mathbb{F}}(\delta(I))=|I|$; in particular, setting $I=\langle\chi(C)\rangle_{\mathbf{P}}$ immediately implies the rest.
\end{proof}

\setlength{\parindent}{2em}
Next, we consider $r$-minimality with respect to different posets.

\begin{corollary}
Let $C\leqslant_{\mathbb{F}}\mathbf{H}$, and let $\mathbf{Q}=(\Omega,\preccurlyeq_{\mathbf{Q}})$ be a poset such that $i\preccurlyeq_{\mathbf{Q}}j$ implies $i\preccurlyeq_{\mathbf{P}}j$ for all $i,j\in\chi(C)$. Then, for $D\leqslant_{\mathbb{F}}C$, if $D$ is $\mathbf{P}$-minimal in $C$, then $D$ is $\mathbf{Q}$-minimal in $C$. Moreover, for $r\in\mathbb{N}$, if $C$ is $r$-minimal with respect to $\mathbf{P}$, then $C$ is $r$-minimal with respect to $\mathbf{Q}$.
\end{corollary}

\begin{proof}
It suffices to prove the first assertion. Let $\alpha\in C$ with $\supp(\alpha)\subseteq\langle\chi(D)\rangle_{\mathbf{Q}}$. For $i\in\supp(\alpha)$, there exists $j\in\chi(D)$ with $i\preccurlyeq_{\mathbf{Q}}j$, which, together with $i,j\in\chi(C)$, implies that $i\preccurlyeq_{\mathbf{P}}j$, and hence $i\in\langle\chi(D)\rangle_{\mathbf{P}}$. It follows that $\supp(\alpha)\subseteq\langle\chi(D)\rangle_{\mathbf{P}}$, which, together with Proposition 3.1, implies that $\alpha\in D$. Now an application of Proposition 3.1 to $\mathbf{Q}$ immediately yields the desired result.
\end{proof}

\setlength{\parindent}{2em}
As an immediate consequence of Corollary 3.2, we derive the following relationship between $r$-minimality with respect to $\mathbf{P}$ and with respect to Hamming support.

\begin{corollary}
Let $C\leqslant_{\mathbb{F}}\mathbf{H}$ and $r\in\mathbb{N}$. If $C$ is $r$-minimal with respect to $\mathbf{P}$, then $C$ is $r$-minimal with respect to Hamming support; conversely, if $\chi(C)$ is an anti-chain of $\mathbf{P}$ and $C$ is $r$-minimal with respect to Hamming support, then $C$ is $r$-minimal with respect to $\mathbf{P}$.
\end{corollary}

\setlength{\parindent}{2em}
Now we show that $r$-minimality is invariant under $\mathbf{P}$-weight isometries.

\setlength{\parindent}{0em}
\begin{corollary}
Let $C\leqslant_{\mathbb{F}}\mathbf{H}$, and let $\varphi\in\End_{\mathbb{F}}(\mathbf{H})$ be a $\mathbf{P}$-weight isometry of $\mathbf{H}$. Then, for $D\leqslant_{\mathbb{F}}C$, if $D$ is $\mathbf{P}$-minimal in $C$, then $\varphi[D]$ is $\mathbf{P}$-minimal in $\varphi[C]$. Moreover, for $r\in\mathbb{N}$, if $C$ is $r$-minimal with respect to $\mathbf{P}$, then $\varphi[C]$ is $r$-minimal with respect to $\mathbf{P}$.
\end{corollary}

\begin{proof}
It suffices to prove the first assertion. By Lemma 2.1, we can choose $\lambda\in\Aut(\mathbf{P})$ such that $\langle\chi(\varphi[A])\rangle_{\mathbf{P}}=\lambda[\langle\chi(A)\rangle_{\mathbf{P}}]$ for all $A\subseteq\mathbf{H}$. For $\theta\in\varphi[C]$ with $\supp(\theta)\subseteq\langle\chi(\varphi[D])\rangle_{\mathbf{P}}$, from $\langle\supp(\theta)\rangle_{\mathbf{P}}=\lambda[\langle\supp(\varphi^{-1}(\theta))\rangle_{\mathbf{P}}]\subseteq\lambda[\langle\chi(D)\rangle_{\mathbf{P}}]$, we deduce that $\langle\supp(\varphi^{-1}(\theta))\rangle_{\mathbf{P}}\subseteq\langle\chi(D)\rangle_{\mathbf{P}}$, which, together with $\varphi^{-1}(\theta)\in C$ and Proposition 3.1, implies that $\varphi^{-1}(\theta)\in D$, which further implies that $\theta\in\varphi[D]$. Now the desired result again follows from Proposition 3.1.
\end{proof}

\setlength{\parindent}{2em}
In the following proposition, we prove the aforementioned fact that $r$-constant weight codes are necessarily $r$-minimal.

\begin{proposition}
Fix $\omega:\Omega\longrightarrow\mathbb{R}^{+}$, and let $C\leqslant_{\mathbb{F}}\mathbf{H}$. For $D\leqslant_{\mathbb{F}}C$, if $\Wt_{(\mathbf{P},\omega)}(D)\leqslant\Wt_{(\mathbf{P},\omega)}(B)$ for all $B\leqslant_{\mathbb{F}}C$ with $\dim_{\mathbb{F}}(B)=\dim_{\mathbb{F}}(D)$, then $D$ is $\mathbf{P}$-minimal in $C$. Moreover, for $r\in\mathbb{N}$, if $C$ is $r$-constant weight with respect to $(\mathbf{P},\omega)$, then $C$ is $r$-minimal with respect to $\mathbf{P}$.
\end{proposition}

\begin{proof}
It suffices to prove the first assertion. Let $B\leqslant_{\mathbb{F}}C$ such that $\dim_{\mathbb{F}}(B)=\dim_{\mathbb{F}}(D)$, $\langle\chi(B)\rangle_{\mathbf{P}}\subseteq\langle\chi(D)\rangle_{\mathbf{P}}$. It follows from $\Wt_{(\mathbf{P},\omega)}(D)\leqslant\Wt_{(\mathbf{P},\omega)}(B)$ and (2.8) that
$$\sum_{i\in\langle\chi(D)\rangle_{\mathbf{P}}}\omega(i)\leqslant\sum_{i\in\langle\chi(B)\rangle_{\mathbf{P}}}\omega(i),$$
which, together with $\langle\chi(B)\rangle_{\mathbf{P}}\subseteq\langle\chi(D)\rangle_{\mathbf{P}}$ and $\omega(i)\in\mathbb{R}^{+}$ for all $i\in\Omega$, immediately implies that $\langle\chi(B)\rangle_{\mathbf{P}}=\langle\chi(D)\rangle_{\mathbf{P}}$, as desired.
\end{proof}

\setlength{\parindent}{2em}
The following proposition is an analogue of Corollary 3.4 for $r$-constant weight codes.

\setlength{\parindent}{0em}
\begin{proposition}
Fix $\omega:\Omega\longrightarrow\mathbb{R}^{+}$. Let $C\leqslant_{\mathbb{F}}\mathbf{H}$ and $r\in\mathbb{N}$ such that $C$ is $r$-constant weight with respect to $(\mathbf{P},\omega)$. Then, for any $(\mathbf{P},\omega)$-weight isometry $\varphi\in\End_{\mathbb{F}}(\mathbf{H})$, $\varphi[C]$ is $r$-constant weight with respect to $(\mathbf{P},\omega)$.
\end{proposition}

\begin{proof}
Let $\sigma\in\mathbb{R}$ such that $\Wt_{(\mathbf{P},\omega)}(D)=\sigma$ for all $D\leqslant_{\mathbb{F}}C$ with $\dim_{\mathbb{F}}(D)=r$. Then, for any $U\leqslant_{\mathbb{F}}\varphi[C]$ with $\dim_{\mathbb{F}}(U)=r$, it follows from $\varphi^{-1}[U]\leqslant_{\mathbb{F}}C$, $\dim_{\mathbb{F}}(\varphi^{-1}[U])=r$ and Lemma 2.1 that $\Wt_{(\mathbf{P},\omega)}(U)=\Wt_{(\mathbf{P},\omega)}(\varphi^{-1}[U])=\sigma$, as desired.
\end{proof}

\subsection{Characterizing $r$-minimal poset codes in terms of cutting $r$-blocking maps}
\setlength{\parindent}{0em}
\begin{lemma}
Let $k\in\mathbb{Z}^{+}$, $G\in\Mat_{k,\Omega}(\mathbb{F})$, $C=\langle\row(G)\rangle_{\mathbb{F}}$, $f:\Omega\longrightarrow\mathbb{F}^{[k]}$ be the column map of $G$. Moreover, let $B,L\subseteq\mathbb{F}^{k}$, and let $D=\{\gamma G\mid \gamma\in B\}$, $M=\{\gamma G\mid \gamma\in L\}$. Then, for $i\in\Omega$, it holds that $$i\not\in\langle\chi(D)\rangle_{\mathbf{P}}\Longleftrightarrow (\forall~j\in\Omega:i\preccurlyeq_{\mathbf{P}}j\Longrightarrow f(j)\in B^{\bot}).$$
Moreover, it holds that
$$\chi(M)\subseteq\langle\chi(D)\rangle_{\mathbf{P}}\Longleftrightarrow\langle\{f(i)\mid i\in\Omega,(\forall~j\in\Omega:i\preccurlyeq_{\mathbf{P}}j\Longrightarrow f(j)\in B^{\bot})\}\rangle_{\mathbb{F}}\subseteq L^{\bot}.$$
\end{lemma}

\begin{proof}
First of all, for any $j\in\Omega$, we infer that $j\not\in\chi(D)$ if and only if $(\gamma G)_{j}=\gamma\cdot f(j)=0$ for all $\gamma\in B$, if and only if $f(j)\in B^{\bot}$. Hence for $i\in\Omega$, $i\not\in\langle\chi(D)\rangle_{\mathbf{P}}$ if and only if $i\not\preccurlyeq_{\mathbf{P}}j$ for all $j\in\chi(D)$, if and only if $i\not\preccurlyeq_{\mathbf{P}}j$ for all $j\in\Omega$ with $f(j)\not\in B^{\bot}$, which establishes the first assertion. From the first assertion, we deduce that $\chi(M)\subseteq\langle\chi(D)\rangle_{\mathbf{P}}$ if and only if $\Omega-\langle\chi(D)\rangle_{\mathbf{P}}\subseteq\Omega-\chi(M)$, if and only if for any $i\in \Omega$ such that $(\forall~j\in \Omega:i\preccurlyeq_{\mathbf{P}}j\Longrightarrow f(j)\in B^{\bot})$, it holds that $f(i)\in L^{\bot}$, which further establishes the second assertion.
\end{proof}

\setlength{\parindent}{2em}
Now we are ready to prove the main result of this subsection.

\setlength{\parindent}{0em}
\begin{theorem}
Let $C\leqslant_{\mathbb{F}}\mathbf{H}$ with $\dim_{\mathbb{F}}(C)=k\geqslant1$, $G\in\Mat_{k,\Omega}(\mathbb{F})$ be a generator matrix of $C$, and let $f:\Omega\longrightarrow\mathbb{F}^{[k]}$ be the column map of $G$. Then, for $B\leqslant_{\mathbb{F}}\mathbb{F}^{k}$ and $D\triangleq\{\gamma G\mid \gamma\in B\}$, $D$ is $\mathbf{P}$-minimal in $C$ if and only if
\begin{equation}\langle\{f(i)\mid i\in\Omega,(\forall~j\in\Omega:i\preccurlyeq_{\mathbf{P}}j\Longrightarrow f(j)\in B^{\bot})\}\rangle_{\mathbb{F}}=B^{\bot}.\end{equation}
Moreover, for $r\in\{0,1,\dots,k\}$, $C$ is $r$-minimal with respect to $\mathbf{P}$ if and only if $f$ is a cutting $r$-blocking map of $\mathbb{F}^{[k]}$ with respect to $\mathbf{P}$.
\end{theorem}

\begin{proof}
First, suppose that (3.1) holds. Let $M\leqslant_{\mathbb{F}}C$ such that $\dim_{\mathbb{F}}(M)=\dim_{\mathbb{F}}(D)$ and $\langle\chi(M)\rangle_{\mathbf{P}}\subseteq\langle\chi(D)\rangle_{\mathbf{P}}$, and let $L\leqslant_{\mathbb{F}}\mathbb{F}^{k}$ such that $M=\{\gamma G\mid \gamma\in L\}$. By Lemma 3.1 and (3.1), we have $B^{\bot}\subseteq L^{\bot}$, which, together with $\dim_{\mathbb{F}}(L)=\dim_{\mathbb{F}}(B)$, implies that $L=B$, which further implies that $M=D$, as desired.

\hspace*{4mm}\,\,Second, suppose that $D$ is $\mathbf{P}$-minimal in $C$. Let $A$ denote the left hand side of (3.1), and suppose by way of contradiction that $A\subsetneqq B^{\bot}$. Noticing that $\langle\col(G)\rangle_{\mathbb{F}}=\mathbb{F}^{[k]}$, we have $B^{\bot}\neq\mathbb{F}^{[k]}$. Hence we can choose $I\leqslant_{\mathbb{F}}\mathbb{F}^{[k]}$ such that $\dim_{\mathbb{F}}(I)=1$ and $I\nsubseteq B^{\bot}$. Since $\dim_{\mathbb{F}}(A+I)\leqslant\dim_{\mathbb{F}}(B^{\bot})$, we can choose $W\leqslant_{\mathbb{F}}\mathbb{F}^{[k]}$ such that $\dim_{\mathbb{F}}(W)=\dim_{\mathbb{F}}(B^{\bot})$ and $A+I\subseteq W$. Let $U\leqslant_{\mathbb{F}}\mathbb{F}^{k}$ such that $W=U^{\bot}$, and let $Q=\{\gamma G\mid\gamma\in U\}\leqslant_{\mathbb{F}}C$. Then, we have $\dim_{\mathbb{F}}(Q)=\dim_{\mathbb{F}}(U)=\dim_{\mathbb{F}}(B)=\dim_{\mathbb{F}}(D)$. By $I\nsubseteq B^{\bot}$ and $I\subseteq W=U^{\bot}$, we have $B\neq U$, which implies that $D\neq Q$. Moreover, from $A\subseteq W=U^{\bot}$ and Lemma 3.1, we deduce that $\langle\chi(Q)\rangle_{\mathbf{P}}\subseteq\langle\chi(D)\rangle_{\mathbf{P}}$. It then follows from Proposition 3.1 that $D$ is not $\mathbf{P}$-minimal in $C$, a contradiction, as desired.

\hspace*{4mm}\,\,Finally, for $r\in\{0,1,\dots,k\}$, the ``moreover'' part follows from the first part and the fact that $\{L^{\bot}\mid L\leqslant_{\mathbb{F}}\mathbb{F}^{k},\dim_{\mathbb{F}}(L)=r\}=\{V\leqslant_{\mathbb{F}}\mathbb{F}^{[k]}\mid\dim_{\mathbb{F}}(V)=k-r\}$, as desired.
\end{proof}

\setlength{\parindent}{2em}
An application of Theorem 3.2 to the anti-chain poset immediately yields the following equivalence between $r$-minimal Hamming metric codes and cutting $r$-blocking sets. Recently, the same result has also been proven in Bishnoi and Tomon \cite{9}, and by Chen, Pan and Xie in \cite{59}.

\begin{corollary}(c.f. [9, Theorem 6] and [59, Theorem 4.2])
Let $C\leqslant_{\mathbb{F}}\mathbf{H}$ with $\dim_{\mathbb{F}}(C)=k\geqslant1$, $G\in\Mat_{k,\Omega}(\mathbb{F})$ be a generator matrix of $C$, and let $S=\col(G)$. Then, for $B\leqslant_{\mathbb{F}}\mathbb{F}^{k}$ and $D\triangleq\{\gamma G\mid \gamma\in B\}$, $D$ is Hamming minimal in $C$ if and only if $\langle S\cap{B^{\bot}}\rangle_{\mathbb{F}}=B^{\bot}$. Moreover, for $r\in\{0,1,\dots,k\}$, $C$ is $r$-minimal with respect to Hamming support if and only if $S$ is a cutting $r$-blocking set of $\mathbb{F}^{[k]}$.
\end{corollary}

\begin{remark}
When $r=1$, Corollary 3.5 recovers [1, Theorem 3.4] and [47, Theorem 3.2] which have been established for minimal Hamming metric codes and cutting blocking sets.
\end{remark}

\setlength{\parindent}{2em}
A combination of Theorem 3.2 and Proposition 2.1 immediately yields the following result.

\begin{corollary}
Let $C\leqslant_{\mathbb{F}}\mathbf{H}$, and let $r\in\{0,1,\dots,\dim_{\mathbb{F}}(C)-1\}$ such that $C$ is $r$-minimal with respect to $\mathbf{P}$. Then, for any $s\in\{0,1,\dots,r\}$, $C$ is $s$-minimal with respect to $\mathbf{P}$.
\end{corollary}

\begin{corollary}
Let $r\in\{0,1,\dots,n-1\}$, $C\leqslant_{\mathbb{F}}\mathbf{H}$ with $\dim_{\mathbb{F}}(C)=r+1$, $G\in\Mat_{r+1,\Omega}(\mathbb{F})$ be a generator matrix of $C$, and let $f:\Omega\longrightarrow\mathbb{F}^{[r+1]}$ be the column map of $G$. Then, $C$ is $r$-minimal with respect to $\mathbf{P}$ if and only if for any $I\leqslant_{\mathbb{F}}\mathbb{F}^{[r+1]}$ with $\dim_{\mathbb{F}}(I)=1$, there exists $j\in\max_{\mathbf{P}}(\chi(C))$ with $f(j)\in I$.
\end{corollary}

\begin{proof}
This follows from Theorem 3.2 and the fact that for $I\leqslant_{\mathbb{F}}\mathbb{F}^{[r+1]}$ with $\dim_{\mathbb{F}}(I)=1$, (3.1) holds true for $B^{\bot}=I$ if and only if there exists $j\in\max_{\mathbf{P}}(\chi(C))$ with $f(j)\in I$.
\end{proof}

\setlength{\parindent}{2em}
We end this subsection with the following proposition in which we characterize $(k-1)$-minimal $k$-dimensional codes in terms of constant weight codes with respect to weighted poset metric.

\begin{proposition}
Let $C\leqslant_{\mathbb{F}}\mathbf{H}$ with $\dim_{\mathbb{F}}(C)=k\geqslant1$. Then, $C$ is $(k-1)$-minimal with respect to $\mathbf{P}$ if and only if there exists $\omega:\Omega\longrightarrow\mathbb{Z}^{+}$ such that $C$ is $(k-1)$-constant weight with respect to $(\mathbf{P},\omega)$.
\end{proposition}

\begin{proof}
Noticing that the ``if'' part follows from Proposition 3.2, we only prove the ``only if'' part. Let $G\in\Mat_{k,\Omega}(\mathbb{F})$ be a generator matrix of $C$, and let $f:\Omega\longrightarrow\mathbb{F}^{[k]}$ be the column map of $G$. Moreover, set $\rho(I)\triangleq\{i\in\langle\chi(C)\rangle_{\mathbf{P}}\mid\forall~j\in\Omega:\,i\preccurlyeq_{\mathbf{P}}j\Longrightarrow f(j)\in I\}$ for all $I\in\Delta$, where $\Delta\triangleq\{L\leqslant_{\mathbb{F}}\mathbb{F}^{[k]}\mid\dim_{\mathbb{F}}(L)=1\}$. By Corollary 3.7, we have $\rho(I)\neq\emptyset$ for all $I\in\Delta$; moreover, it is straightforward to verify that $\rho(I)\cap\rho(J)=\emptyset$ for all $I\neq J\in\Delta$. Hence we can choose $\sigma\in\mathbb{Z}^{+}$ and $\omega:\Omega\longrightarrow\mathbb{Z}^{+}$ such that $\sum_{i\in\rho(I)}\omega(i)=\sigma$ for all $I\in\Delta$. Now let $D\leqslant_{\mathbb{F}}C$ with $\dim_{\mathbb{F}}(D)=k-1$, and let $B\leqslant_{\mathbb{F}}\mathbb{F}^{k}$ such that $D=\{\gamma G\mid\gamma\in B\}$. It then follows from Lemma 3.1 and $B^{\bot}\in\Delta$ that $\langle\chi(C)\rangle_{\mathbf{P}}-\langle\chi(D)\rangle_{\mathbf{P}}=\rho(B^{\bot})$, which, together with (2.8), implies that $\Wt_{(\mathbf{P},\omega)}(D)=\Wt_{(\mathbf{P},\omega)}(C)-\sigma$, as desired.
\end{proof}

\subsection{Characterizing $r$-minimal poset codes in terms of $(\mathbf{P},\omega)$-weight}

\setlength{\parindent}{2em}
Throughout this subsection, we fix $\omega:\Omega\longrightarrow\mathbb{R}^{+}$ and consider the weighted poset $(\mathbf{P},\omega)$.

\setlength{\parindent}{0em}
\begin{lemma}
Let $D\leqslant_{\mathbb{F}}\mathbf{H}$ with $\dim_{\mathbb{F}}(D)=r$, $\alpha\in\mathbf{H}$, and let $J=\langle\chi(D)\rangle_{\mathbf{P}}-\left(\bigcap_{\beta\in D}\langle\supp(\alpha-\beta)\rangle_{\mathbf{P}}\right)$. Then, it holds that
\begin{eqnarray*}
\begin{split}
\sum_{\beta\in D}\wt_{(\mathbf{P},\omega)}(\alpha-\beta)&=q^{r}\Wt_{(\mathbf{P},\omega)}(D)+q^{r}\left(\sum_{i\in\langle\supp(\alpha)\rangle_{\mathbf{P}}-\langle\chi(D)\rangle_{\mathbf{P}}}\omega(i)\right)-\left(\sum_{i\in J}|\{\beta\in D\mid \beta\mid_{\langle\{i\}\rangle_{\overline{\mathbf{P}}}}=0\}|\omega(i)\right)\\
&\geqslant q^{r}\Wt_{(\mathbf{P},\omega)}(D)-\left(\sum_{i\in \langle\chi(D)\rangle_{\mathbf{P}}}|\{\beta\in D\mid \beta\mid_{\langle\{i\}\rangle_{\overline{\mathbf{P}}}}=0\}|\omega(i)\right)\\
&=\sum_{\beta\in D-\{0\}}\wt_{(\mathbf{P},\omega)}(\beta),
\end{split}
\end{eqnarray*}
where for any $\beta\in \mathbf{H}$ and $I\subseteq \Omega$, $\beta\mid_{I}\triangleq(\beta_i\mid i\in I)\in\mathbb{F}^{I}$. Moreover, the second ``$\geqslant$'' becomes equlity if and only if $\supp(\alpha)\subseteq\langle\chi(D)\rangle_{\mathbf{P}}$ and $\bigcap_{\beta\in D}\langle\supp(\alpha-\beta)\rangle_{\mathbf{P}}=\emptyset$.
\end{lemma}

\begin{proof}
For $I\triangleq\bigcap_{\beta\in D}\langle\supp(\alpha-\beta)\rangle_{\mathbf{P}}$, one can check that $\langle\supp(\alpha)\rangle_{\mathbf{P}}-I\subseteq\langle\chi(D)\rangle_{\mathbf{P}}$, which further implies that $(\langle\supp(\alpha)\rangle_{\mathbf{P}}\cup\langle\chi(D)\rangle_{\mathbf{P}})-I=\langle\chi(D)\rangle_{\mathbf{P}}-I=J$. For $i\in \Omega$ and $\beta\in D$, we infer that $i\not\in\langle\supp(\alpha-\beta)\rangle_{\mathbf{P}}\Longleftrightarrow\beta\mid_{\langle\{i\}\rangle_{\overline{\mathbf{P}}}}=\alpha\mid_{\langle\{i\}\rangle_{\overline{\mathbf{P}}}}$. Hence for any $i\in\Omega-I$, there exists $\beta\in D$ such that $\beta\mid_{\langle\{i\}\rangle_{\overline{\mathbf{P}}}}=\alpha\mid_{\langle\{i\}\rangle_{\overline{\mathbf{P}}}}$, which further implies that
$$|\{\beta\in D\mid i\not\in\langle\supp(\alpha-\beta)\rangle_{\mathbf{P}}\}=|\{\beta\in D\mid \beta\mid_{\langle\{i\}\rangle_{\overline{\mathbf{P}}}}=\alpha\mid_{\langle\{i\}\rangle_{\overline{\mathbf{P}}}}\}|=|\{\beta\in D\mid \beta\mid_{\langle\{i\}\rangle_{\overline{\mathbf{P}}}}=0\}|.$$
Now let $v$ denote the left hand side of the equation of the lemma. Then, we have
\begin{eqnarray*}
\begin{split}
v&=\sum_{\beta\in D}\sum_{i\in\langle\supp(\alpha-\beta)\rangle_{\mathbf{P}}}\omega(i)\\
&=\sum_{i\in\langle\supp(\alpha)\rangle_{\mathbf{P}}\cup\langle\chi(D)\rangle_{\mathbf{P}}}|\{\beta\in D\mid i\in\langle\supp(\alpha-\beta)\rangle_{\mathbf{P}}\}|\omega(i)\\
&=\left(\sum_{i\in\langle\supp(\alpha)\rangle_{\mathbf{P}}\cup\langle\chi(D)\rangle_{\mathbf{P}}}q^{r}\omega(i)\right)-\left(\sum_{i\in\langle\supp(\alpha)\rangle_{\mathbf{P}}\cup\langle\chi(D)\rangle_{\mathbf{P}}}|\{\beta\in D\mid i\not\in\langle\supp(\alpha-\beta)\rangle_{\mathbf{P}}\}|\omega(i)\right)\\
&=q^{r}\Wt_{(\mathbf{P},\omega)}(D)+q^{r}\left(\sum_{i\in\langle\supp(\alpha)\rangle_{\mathbf{P}}-\langle\chi(D)\rangle_{\mathbf{P}}}\omega(i)\right)-\left(\sum_{i\in J}|\{\beta\in D\mid \beta\mid_{\langle\{i\}\rangle_{\overline{\mathbf{P}}}}=0\}|\omega(i)\right),
\end{split}
\end{eqnarray*}
which establishes the first equality, as desired. Moreover, the third equality follows from the first equality with $\alpha$ set to be $0$, and the second inequality follows from the facts $\omega(i)>0$ for all $i\in\Omega$ and $J\subseteq\langle\chi(D)\rangle_{\mathbf{P}}$. Since $\omega(i)>0$ and $|\{\beta\in D\mid \beta\mid_{\langle\{i\}\rangle_{\overline{\mathbf{P}}}}=0\}|\geqslant1$ for all $i\in\Omega$, the second ``$\geqslant$'' becomes equality if and only if $\langle\supp(\alpha)\rangle_{\mathbf{P}}-\langle\chi(D)\rangle_{\mathbf{P}}=\emptyset$ and $J=\langle\chi(D)\rangle_{\mathbf{P}}$, if and only if $\supp(\alpha)\subseteq\langle\chi(D)\rangle_{\mathbf{P}}$ and $I=\emptyset$, as desired.
\end{proof}

\setlength{\parindent}{2em}
Now we derive necessary and sufficient conditions for a subcode to be $\mathbf{P}$-minimal.

\setlength{\parindent}{0em}
\begin{proposition}
Let $C\leqslant_{\mathbb{F}}\mathbf{H}$, and let $D\leqslant_{\mathbb{F}}C$ with $\dim_{\mathbb{F}}(D)=r$. Then, $D$ is $\mathbf{P}$-minimal in $C$ if and only if for any $\alpha\in C-D$, the following two conditions hold:
\begin{equation}\supp(\alpha)\subseteq\langle\chi(D)\rangle_{\mathbf{P}}\Longrightarrow\bigcap_{\beta\in D}\langle\supp(\alpha-\beta)\rangle_{\mathbf{P}}=\emptyset,\end{equation}
\begin{equation}\sum_{\beta\in D}\wt_{(\mathbf{P},\omega)}(\alpha-\beta)\neq q^{r}\Wt_{(\mathbf{P},\omega)}(D)-\left(\sum_{i\in \langle\chi(D)\rangle_{\mathbf{P}}}|\{\beta\in D\mid \beta\mid_{\langle\{i\}\rangle_{\overline{\mathbf{P}}}}=0\}|\omega(i)\right).\end{equation}
In particular, if (3.2) holds for all $\alpha\in C-D$ and
\begin{equation}\frac{\wt_{(\mathbf{P},\omega)}(\alpha)}{\wt_{(\mathbf{P},\omega)}(\theta)}>1-q^{-r},\end{equation}
for all $\alpha\in C-D$ and $\theta\in D-\{0\}$, then $D$ is $\mathbf{P}$-minimal in $C$.
\end{proposition}

\begin{proof}
By Proposition 3.1, $D$ is $\mathbf{P}$-minimal in $C$ if and only if $\supp(\alpha)\nsubseteq\langle\chi(D)\rangle_{\mathbf{P}}$ for all $\alpha\in C-D$. This, together with Lemma 3.2, immediately establishes the first assertion. Now we prove the ``in particular'' part. Let $a\triangleq\max\{\wt_{(\mathbf{P},\omega)}(\beta)\mid\beta\in D\}$. Then, for $\alpha\in C-D$ and $b\triangleq\min\{\wt_{(\mathbf{P},\omega)}(\alpha-\beta)\mid\beta\in D\}$, it follows from (3.4) that $b>(1-q^{-r})a$, which implies that
$$\sum_{\beta\in D}\wt_{(\mathbf{P},\omega)}(\alpha-\beta)\geqslant q^{r}b>(q^{r}-1)a\geqslant\sum_{\beta\in D-\{0\}}\wt_{(\mathbf{P},\omega)}(\beta).$$
It then follows from Lemma 3.2 that (3.3) holds true for all $\alpha\in C-D$, as desired.
\end{proof}

\setlength{\parindent}{2em}
Now we state the main result of this subsection, in which we give necessary and sufficient conditions for a code to be $r$-minimal in terms of $(\mathbf{P},\omega)$-weight of its codewords and subcodes. The following theorem immediately follows from Proposition 3.5.

\setlength{\parindent}{0em}
\begin{theorem}
Let $\{0\}\neq C\leqslant_{\mathbb{F}}\mathbf{H}$ and $r\in\mathbb{N}$. Then, $C$ is $r$-minimal with respect to $\mathbf{P}$ if and only if for any $D\leqslant_{\mathbb{F}}C$ with $\dim_{\mathbb{F}}(D)=r$ and any $\alpha\in C-D$, the following two conditions hold:
\begin{equation}\supp(\alpha)\subseteq\langle\chi(D)\rangle_{\mathbf{P}}\Longrightarrow\bigcap_{\beta\in D}\langle\supp(\alpha-\beta)\rangle_{\mathbf{P}}=\emptyset,\end{equation}
\begin{equation}\sum_{\beta\in D}\wt_{(\mathbf{P},\omega)}(\alpha-\beta)\neq q^{r}\Wt_{(\mathbf{P},\omega)}(D)-\left(\sum_{i\in \langle\chi(D)\rangle_{\mathbf{P}}}|\{\beta\in D\mid \beta\mid_{\langle\{i\}\rangle_{\overline{\mathbf{P}}}}=0\}|\omega(i)\right).\end{equation}
In particular, if (3.5) always holds and
\begin{equation}\frac{\min\{\wt_{(\mathbf{P},\omega)}(\alpha)\mid\alpha\in C-\{0\}\}}{\max\{\wt_{(\mathbf{P},\omega)}(\alpha)\mid\alpha\in C\}}>1-q^{-r},\end{equation}
then $C$ is $r$-minimal with respect to $\mathbf{P}$.
\end{theorem}

\begin{remark}
If $\mathbf{P}$ is the anti-chain poset, then for $D\leqslant_{\mathbb{F}}\mathbf{H}$ and $\alpha\in\mathbf{H}$, we have $\chi(D)\cap\bigcap_{\beta\in D}\supp(\alpha-\beta)=\emptyset$ and $\supp(\alpha)\subseteq\chi(D)\Longrightarrow\bigcap_{\beta\in D}\supp(\alpha-\beta)=\emptyset$; moreover, $|\{\beta\in D\mid \beta_i=0\}|=|D|/q$ for all $i\in \chi(D)$. Hence when $r=1$ and $\omega(i)=1$ for all $i\in\Omega$, (3.6) boils down to (1.3), and hence Theorem 3.3 generalizes [25, Theorem 11]; moreover, (3.7) boils down to the Ashikhmin-Barg condition (1.2).
\end{remark}

\subsection{Existence results}
\setlength{\parindent}{2em}
In this subsection, we give existence results for $r$-minimal codes with respect to $\mathbf{P}$, and we begin by computing the number of $(r+1)$-dimensional $r$-minimal codes, as detailed in the following proposition.

\begin{proposition}
Let $r\in\{0,1,\dots,n-1\}$, and let $\theta\triangleq(q^{r+1}-1)/(q-1)$. Then, for $I\in\mathcal{I}(\mathbf{P})$ and $J=\max_{\mathbf{P}}(I)$, the number of all the $(r+1)$-dimensional $r$-minimal codes with respect to $\mathbf{P}$ with $\mathbf{P}$-support $I$ is equal to
\begin{equation}\sum_{t=0}^{\theta}(-1)^{t}\binom{\theta}{t}(q^{r+1}-1-t(q-1))^{|J|} q^{(r+1)|I-J|}\left(\prod_{i=0}^{r}(q^{r+1}-q^{i})\right)^{-1}.\end{equation}
Moreover, the number of all the $(r+1)$-dimensional $r$-minimal codes with respect to $\mathbf{P}$ is equal to
\begin{equation}\sum_{t=0}^{\theta}(-1)^{t}\binom{\theta}{t}\left(\sum_{I\in\mathcal{I}(\mathbf{P})}(q^{r+1}-1-t(q-1))^{|\max_{\mathbf{P}}(I)|} q^{(r+1)|I-\max_{\mathbf{P}}(I)|}\right)\left(\prod_{i=0}^{r}(q^{r+1}-q^{i})\right)^{-1}.\end{equation}
\end{proposition}

\begin{proof}
First, we note that for any $K\subseteq\Omega$, it holds that $\langle K\rangle_{\mathbf{P}}=I\Longleftrightarrow J\subseteq K\subseteq I$; moreover, if $J\subseteq K\subseteq I$, then we have $\max_{\mathbf{P}}(K)=J$. Now let $K\subseteq\Omega$ with $J\subseteq K\subseteq I$, and let $a$ denote the number of all the $(r+1)$-dimensional $r$-minimal codes with respect to $\mathbf{P}$ whose Hamming support is equal to $K$. We first compute $a$. Let $\Delta\triangleq\{B\leqslant_{\mathbb{F}}\mathbb{F}^{[r+1]}\mid\dim_{\mathbb{F}}(B)=1\}$, and define $\rho:\mathbb{F}^{[r+1]}-\{0\}\longrightarrow\Delta$ as $\rho(\theta)=\langle\{\theta\}\rangle_{\mathbb{F}}$. We note that every $(r+1)$-dimensional $\mathbb{F}$-subspace of $\mathbf{H}$ has $\prod_{i=0}^{r}(q^{r+1}-q^{i})$ generator matrices in $\Mat_{r+1,\Omega}(\mathbb{F})$. Moreover, consider an arbitrary $f:\Omega\longrightarrow\mathbb{F}^{[r+1]}$. Let $G\in\Mat_{r+1,\Omega}(\mathbb{F})$ such that $f$ is equal to the column map of $G$, and let $C=\langle\row(G)\rangle_{\mathbb{F}}$. By Lemma 3.1 and Corollary 3.7, $C$ is an $(r+1)$-dimensional $r$-minimal code with respect to $\mathbf{P}$ such that $\chi(C)=K$ if and only if both $\{i\in\Omega\mid f(i)\neq0\}=K$ and $\ran(\rho\circ f\mid_{J})=\Delta$ hold true. The above discussion implies that
\begin{eqnarray*}
\begin{split}
\left(\prod_{i=0}^{r}(q^{r+1}-q^{i})\right)a&=|\{f:\Omega\longrightarrow\mathbb{F}^{[r+1]}\mid \{i\in\Omega\mid f(i)\neq0\}=K,\ran(\rho\circ f\mid_{J})=\Delta\}|\\
&=(q^{r+1}-1)^{|K-J|}|\{\eta:J\longrightarrow\mathbb{F}^{[r+1]}-\{0\}\mid \ran(\rho\circ \eta)=\Delta\}|\\
&=(q^{r+1}-1)^{|K-J|}(q-1)^{|J|}|\{T:J\longrightarrow\Delta\mid \ran(T)=\Delta\}|\\
&=(q^{r+1}-1)^{|K-J|}(q-1)^{|J|}\left(\sum_{t=0}^{\theta}(-1)^{t}\binom{\theta}{t}\left(\theta-t\right)^{|J|}\right)\\
&=\sum_{t=0}^{\theta}(-1)^{t}\binom{\theta}{t}(q^{r+1}-1-t(q-1))^{|J|}(q^{r+1}-1)^{|K-J|}.
\end{split}
\end{eqnarray*}
It then follows that
\begin{eqnarray*}
\begin{split}
&|\{C\leqslant_{\mathbb{F}}\mathbf{H}\mid\text{$\dim_{\mathbb{F}}(C)=r+1$, $C$ is $r$-minimal with respect to $\mathbf{P}$, $\langle\chi(C)\rangle_{\mathbf{P}}=I$}\}|\\
&=\sum_{(K,J\subseteq K\subseteq I)}\sum_{t=0}^{\theta}(-1)^{t}\binom{\theta}{t}(q^{r+1}-1-t(q-1))^{|J|}(q^{r+1}-1)^{|K-J|}\left(\prod_{i=0}^{r}(q^{r+1}-q^{i})\right)^{-1}\\
&=\sum_{t=0}^{\theta}(-1)^{t}\binom{\theta}{t}(q^{r+1}-1-t(q-1))^{|J|}\left(\sum_{(K,J\subseteq K\subseteq I)}(q^{r+1}-1)^{|K-J|}\right)\left(\prod_{i=0}^{r}(q^{r+1}-q^{i})\right)^{-1}\\
&=\sum_{t=0}^{\theta}(-1)^{t}\binom{\theta}{t}(q^{r+1}-1-t(q-1))^{|J|}q^{(r+1)|I-J|}\left(\prod_{i=0}^{r}(q^{r+1}-q^{i})\right)^{-1},
\end{split}
\end{eqnarray*}
establishing (3.8). Moreover, (3.9) immediately follows from (3.8), as desired.
\end{proof}

\setlength{\parindent}{2em}
Now we state and prove the main result of this subsection.

\setlength{\parindent}{0em}
\begin{theorem}
For $(k,r)\in\mathbb{N}^{2}$ with $r+1\leqslant k\leqslant n$ and $\theta\triangleq(q^{r+1}-1)/(q-1)$, if
\begin{equation}
\begin{aligned}
&\left(\prod_{i=0}^{r}(q^{n}-q^{i})\right)-\left(\prod_{i=k-r}^{k}\frac{q^{i+n-k}-1}{q^{i}-1}\right)\left(\prod_{i=0}^{r}(q^{r+1}-q^{i})\right)\\
&<\sum_{t=0}^{\theta}(-1)^{t}\binom{\theta}{t}\left(\sum_{I\in\mathcal{I}(\mathbf{P})}(q^{r+1}-1-t(q-1))^{|\max_{\mathbf{P}}(I)|} q^{(r+1)|I-\max_{\mathbf{P}}(I)|}\right),
\end{aligned}
\end{equation}
then there exists $C\leqslant_{\mathbb{F}}\mathbf{H}$ such that $\dim_{\mathbb{F}}(C)=k$ and $C$ is $r$-minimal with respect to $\mathbf{P}$.
\end{theorem}

\begin{proof}
For any $(b,a)\in\mathbb{N}^{2}$ with $b\geqslant a$, define
$$\bin_{q}(b,a)\triangleq\prod_{i=1}^{a}\frac{q^{i+b-a}-1}{q^{i}-1}.$$
It is well-known that $\bin_{q}(b,a)$ is equal to the number of all the $a$-dimensional $\mathbb{F}$-subspaces of $\mathbb{F}^{b}$. Now let $T$ denote the following binary relation:
$$\{(A,B)\mid B\leqslant_{\mathbb{F}}\mathbf{H},A\leqslant_{\mathbb{F}}B,\dim_{\mathbb{F}}(A)=r+1,\dim_{\mathbb{F}}(B)=k,\text{$A$ is not $r$-minimal with respect to $\mathbf{P}$}\}.$$
From Theorem 3.1, we deduce that
$$\dom(T)=\{A\leqslant_{\mathbb{F}}\mathbf{H}\mid \dim_{\mathbb{F}}(A)=r+1,\text{$A$ is not $r$-minimal with respect to $\mathbf{P}$}\},$$
$$\ran(T)=\{B\leqslant_{\mathbb{F}}\mathbf{H}\mid \dim_{\mathbb{F}}(B)=k,\text{$B$ is not $r$-minimal with respect to $\mathbf{P}$}\}.$$
For notational simplicity, let $\varepsilon$ denote the right hand side of (3.10). For any $A\in\dom(T)$, we have $|\{B\mid(A,B)\in T\}|=\bin_{q}(n-(r+1),k-(r+1))$. This, together with Proposition 3.6, further implies that
\begin{eqnarray*}
\begin{split}
|T|&=\bin_{q}(n-(r+1),k-(r+1))|\dom(T)|\\
&=\bin_{q}(n-r-1,k-r-1)\left(\bin_{q}(n,r+1)-\varepsilon\left(\prod_{i=0}^{r}(q^{r+1}-q^{i})\right)^{-1}\right)\\
&=\bin_{q}(n-r-1,k-r-1)\left(\left(\prod_{i=0}^{r}(q^{n}-q^{i})\right)-\varepsilon\right)\left(\prod_{i=0}^{r}(q^{r+1}-q^{i})\right)^{-1}\\
&<\bin_{q}(n-r-1,k-r-1)\left(\prod_{i=k-r}^{k}\frac{q^{i+n-k}-1}{q^{i}-1}\right)\left(\prod_{i=0}^{r}(q^{r+1}-q^{i})\right)\left(\prod_{i=0}^{r}(q^{r+1}-q^{i})\right)^{-1}\\
&=\bin_{q}(n,k).
\end{split}
\end{eqnarray*}
Therefore we have $|\ran(T)|\leqslant|T|<\bin_{q}(n,k)$, which further establishes the desired result.
\end{proof}

\setlength{\parindent}{2em}
An application of Theorem 3.4 to NRT poset yields the following corollary.

\setlength{\parindent}{0em}
\begin{corollary}
Let $(s,m)\in\mathbb{Z}^{+}\times\mathbb{Z}^{+}$, and suppose that $\mathbf{P}=(\{1,\dots,s\}\times\{1,\dots,m\},\preccurlyeq_{_{NRT}})$. Let $(k,r)\in\mathbb{N}^{2}$ with $r+1\leqslant k\leqslant n$. Suppose that
\begin{equation}s\geqslant\frac{(r+1)(k-r)-\log_{q}(q-1)-\log_{q}\left(\prod_{i=1}^{r}(1-q^{-i})\right)}{(r+1)m-\log_{q}\left(q^{(r+1)m}-\frac{(q-1)(q^{(r+1)m}-1)}{q^{r+1}-1}\right)}.\end{equation}
Then, there exists $C\leqslant_{\mathbb{F}}\mathbf{H}$ such that $\dim_{\mathbb{F}}(C)=k$, $C$ is $r$-minimal with respect to $\mathbf{P}$.
\end{corollary}

\begin{proof}
Set $\Gamma\triangleq\{0,1,\dots,m\}^{s}$. Then, $\varepsilon(a_1,\dots,a_s)\triangleq\bigcup_{i=1}^{s}\{i\}\times\{1,\dots,a_{i}\}$ induces a bijection between $\Gamma$ and $\mathcal{I}(\mathbf{P})$; moreover, we have $\max_{\mathbf{P}}(\varepsilon(a_1,\dots,a_s))=\{(i,a_i)\mid i\in\{1,\dots,s\},a_i\geqslant1\}$ for all $(a_1,\dots,a_s)\in\Gamma$. Write $\theta\triangleq(q^{r+1}-1)/(q-1)$; moreover, for any $t\in\{0,1,\dots,\theta\}$, set
\begin{eqnarray*}
\begin{split}
\sigma_t&\triangleq\binom{\theta}{t}\left(\sum_{I\in\mathcal{I}(\mathbf{P})}(q^{r+1}-1-t(q-1))^{|\max_{\mathbf{P}}(I)|} q^{(r+1)|I-\max_{\mathbf{P}}(I)|}\right)\\
&=\binom{\theta}{t}\left(\sum_{(a_1,\dots,a_s)\in\Gamma}(q^{r+1}-1-t(q-1))^{|\{i\in\{1,\dots,s\}\mid a_i\geqslant1\}|}q^{(r+1)(\sum_{(i\in\{1,\dots,s\},a_i\geqslant1)}(a_i-1))}\right)\\
&=\binom{\theta}{t}\left(\left(\sum_{h=1}^{m}(q^{r+1}-1-t(q-1))q^{(r+1)(h-1)}\right)+1\right)^{s}\\
&=\binom{\theta}{t}\left(q^{(r+1)m}-\frac{t(q-1)(q^{(r+1)m}-1)}{q^{r+1}-1}\right)^{s}.
\end{split}
\end{eqnarray*}
It is straightforward to verify that
\begin{equation}\left(\prod_{i=0}^{r}(q^{n}-q^{i})\right)-\left(\prod_{i=k-r}^{k}\frac{q^{i+n-k}-1}{q^{i}-1}\right)\left(\prod_{i=0}^{r}(q^{r+1}-q^{i})\right)<q^{n(r+1)}-q^{(n-k)(r+1)}\left(\prod_{i=0}^{r}(q^{r+1}-q^{i})\right).\end{equation}
From (3.12) together with some straightforward computation, we deduce that
\begin{equation}q^{n(r+1)}-q^{(n-k)(r+1)}\left(\prod_{i=0}^{r}(q^{r+1}-q^{i})\right)\leqslant\sigma_0-\sigma_1.\end{equation}
We infer that the left hand side of (3.13) is positive, which implies that $\sigma_0>\sigma_1$; moreover, for any $t\in\{0,1,\dots,\theta-1\}$, it is straightforward to check that $1>\frac{\sigma_1}{\sigma_0}\geqslant\frac{\sigma_{t+1}}{\sigma_t}$, which implies that $\sigma_t>\sigma_{t+1}$. Hence the sequence $(\sigma_0,\sigma_1,\dots,\sigma_\theta)$ is decreasing with $\sigma_\theta>0$, which implies that
\begin{equation}\sigma_0-\sigma_1\leqslant\sum_{t=0}^{\theta}(-1)^{t}\sigma_t.\end{equation}
Hence with Theorem 3.4, the desired result immediately follows from (3.12)--(3.14).
\end{proof}

\setlength{\parindent}{2em}
Proposition 3.6 and Corollary 3.8 can be simplified for Hamming support, as detailed in the following example.

\setlength{\parindent}{0em}
\begin{example}
For $r\in\{0,1,\dots,n-1\}$ and $\theta\triangleq(q^{r+1}-1)/(q-1)$, by Proposition 3.6, the number of all the $(r+1)$-dimensional $r$-minimal codes with respect to Hamming support is equal to
$$\sum_{t=0}^{\theta}(-1)^{t}\binom{\theta}{t}\left(q^{r+1}-t(q-1)\right)^{n}\left(\prod_{i=0}^{r}(q^{r+1}-q^{i})\right)^{-1}.$$
Moreover, let $(k,r)\in\mathbb{N}^{2}$ such that $r+1\leqslant k\leqslant n$ and
\begin{equation}n\geqslant\frac{(r+1)(k-r)-\log_{q}(q-1)-\log_{q}\left(\prod_{i=1}^{r}(1-q^{-i})\right)}{r+1-\log_{q}(q^{r+1}-q+1)}.\end{equation}
Then, by Corollary 3.8, there exists $C\leqslant_{\mathbb{F}}\mathbf{H}$ such that $\dim_{\mathbb{F}}(C)=k$ and $C$ is $r$-minimal with respect to Hamming support. When $q=2$, (3.15) seems to provide the best upper bound for the minimal length of $k$-dimensional $r$-minimal binary codes so far (c.f. \cite{3,8,14,23}); moreover, if $q=2$ and $r=1$, then (3.15) recovers [23, Corollary 33]. When $q\geqslant3$, better bounds which are both linear in $k$ and $q$ have been derived in \cite{3,23} for $r=1$ and in \cite{8} for general $r$, and we refer the reader to [3, Theorem 4.10], [8, Theorem 1.1] and [23, Theorems 5, 31] for more details.
\end{example}

\section{$r$-Minimal codes with respect to hierarchical posets}
\setlength{\parindent}{2em}
In this section, we characterize $r$-minimal codes with respect to hierarchical posets in terms of $r$-minimal Hamming metric codes. We begin with the following proposition which holds true for any poset.

\setlength{\parindent}{0em}
\begin{proposition}
Let $\mathbf{Q}=(\Omega,\preccurlyeq_{\mathbf{Q}})$ be a poset, $C\leqslant_{\mathbb{F}}\mathbf{H}$ such that $\chi(C)$ is an anti-chain of $\mathbf{Q}$, and let $r\in\mathbb{N}$ such that $C$ is $r$-minimal with respect to Hamming support. Then, for any $\mathbf{Q}$-weight isometry $\varphi\in\End_{\mathbb{F}}(\mathbf{H})$, $\varphi[C]$ is $r$-minimal with respect to $\mathbf{Q}$.
\end{proposition}

\begin{proof}
It follows from Corollary 3.3 that $C$ is $r$-minimal with respect to $\mathbf{Q}$, which, together with Corollary 3.4, immediately implies the desired result.
\end{proof}

\setlength{\parindent}{2em}
Now we state and prove the main result of this section. The following theorem can be regarded as the converse of Proposition 4.1 for hierarchical posets.

\begin{theorem}
Let $\mathbf{P}=(\Omega,\preccurlyeq_{\mathbf{P}})$ be a hierarchical poset, $C\leqslant_{\mathbb{F}}\mathbf{H}$ with $\dim_{\mathbb{F}}(C)\geqslant2$, and let $r\in\{1,\dots,\dim_{\mathbb{F}}(C)-1\}$ such that $C$ is $r$-minimal with respect to $\mathbf{P}$. Then, there exist a $\mathbf{P}$-support isometry $\varphi\in\End_{\mathbb{F}}(\mathbf{H})$ and $U\leqslant_{\mathbb{F}}\mathbf{H}$ such that $\chi(U)$ is an anti-chain of $\mathbf{P}$, $U$ is $r$-minimal with respect to Hamming support and $C=\varphi[U]$.
\end{theorem}

\begin{proof}
Let $W_{t}=\{u\in\Omega\mid\len_{\mathbf{P}}(u)=t\}$ for all $t\in\mathbb{N}$, and let $m$ denote the largest cardinality of a chain in $\mathbf{P}$. Since $\mathbf{P}$ is hierarchical, by the canonical decomposition [35, Theorem 6], there exists a $\mathbf{P}$-support isometry $\psi\in\End_{\mathbb{F}}(\mathbf{H})$ and a tuple $(U_1,\dots,U_m)$ such that $U_t\leqslant_{\mathbb{F}}\delta(W_t)$ for all $t\in\{1,\dots,m\}$, and that $\psi[C]=U_1+\cdots+U_m$. By Corollary 3.4, $\psi[C]$ is $r$-minimal with respect to $\mathbf{P}$. Noticing that $1\leqslant r<\dim_{\mathbb{F}}(C)=\dim_{\mathbb{F}}(\psi[C])$, by Corollary 3.6, $\psi[C]$ is $1$-minimal with respect to $\mathbf{P}$. Now suppose by way of contradiction that there exists $s,l\in\{1,\dots,m\}$ such that $s+1\leqslant l$, $U_s\neq\{0\}$, $U_l\neq\{0\}$. Then, we can choose $I\leqslant_{\mathbb{F}}U_s$, $J\leqslant_{\mathbb{F}}U_l$ such that $\dim_{\mathbb{F}}(I)=\dim_{\mathbb{F}}(J)=1$. From $I\subseteq\delta(W_s)$, we deduce that $\langle\chi(I)\rangle_{\mathbf{P}}\subseteq\bigcup_{t=1}^{s}W_t\subseteq\bigcup_{t=1}^{l-1}W_t$. From $J\subseteq\delta(W_l)$ and $J\neq\{0\}$, we deduce that $\emptyset\neq\chi(J)\subseteq W_l$, which further implies that $\langle\chi(J)\rangle_{\mathbf{P}}=(\bigcup_{t=1}^{l-1}W_t)\cup\chi(J)$ as $\mathbf{P}$ is hierarchical. It follows that $\langle\chi(I)\rangle_{\mathbf{P}}\subsetneqq\langle\chi(J)\rangle_{\mathbf{P}}$, a contradiction to the $1$-minimality of $\psi[C]$, as desired. Hence there uniquely exists $l\in\{1,\dots,m\}$ with $U_l\neq\{0\}$. Therefore we have $\psi[C]=U_l$ and $C=\psi^{-1}[U_l]$. Since $\psi$ is a $\mathbf{P}$-support isometry of $\mathbf{H}$, $\psi^{-1}$ is a $\mathbf{P}$-support isometry of $\mathbf{H}$. Since $U_l\subseteq\delta(W_l)$, we have $\chi(U_l)\subseteq W_l$, and hence $\chi(U_l)$ is an anti-chain of $\mathbf{P}$. Finally, we have shown that $U_l=\psi[C]$ is $r$-minimal with respect to $\mathbf{P}$, which, together with Corollary 3.3, implies that $U_l$ is $r$-minimal with respect to Hamming support, as desired.
\end{proof}

\section{Cutting $r$-blocking sets induced by hierarchical posets}

\setlength{\parindent}{2em}
We begin with two general constructions that are frequently used in the study of minimal codes. More precisely, for any $M\subseteq\mathbf{H}$ and $g:\mathbf{H}\longrightarrow\mathbb{F}$, set
$$\mathcal{C}(M)=\{(\langle\tau,\theta\rangle\mid\theta\in M)\mid\tau\in\mathbf{H}\}\subseteq\mathbb{F}^{M},$$
$$\mathcal{C^{'}}(g)=\{(\langle\tau,\theta\rangle+bg(\theta)\mid\theta\in \mathbf{H})\mid\tau\in\mathbf{H},b\in\mathbb{F}\}\subseteq\mathbb{F}^{\mathbf{H}}.$$
Minimal codes induced by various $M$ and $g$ have been extensively explored (see, e.g., \cite{10,15,17,18,19,25,28,31,32,37,38,46,53,58}). In this section, we follow \cite{28,53} and study the minimality of $\mathcal{C}(M)$ and $\mathcal{C^{'}}(g)$, where $M$ and $g$ are induced by hierarchical posets with two levels. By Corollary 3.5, for $r\in\{0,1,\dots,n-1\}$, $\mathcal{C}(M)$ is an $n$-dimensional $r$-minimal code with respect to Hamming support if and only if $M$ is a cutting $r$-blocking set of $\mathbf{H}$; similarly, $\mathcal{C^{'}}(g)$ is an $(n+1)$-dimensional $1$-minimal code with respect to Hamming support if and only if $g=\{(\theta,g(\theta))\mid \theta\in\mathbf{H}\}$ is a cutting $1$-blocking set of $\mathbf{H}\times\mathbb{F}$. So, from now on, we study the minimality of $\mathcal{C}(M)$ and $\mathcal{C^{'}}(g)$ by considering whether $M$ and $g$ are cutting ($r$-)blocking sets.

Throughout the rest of this section, let $m\in\{1,\dots,n-1\}$, $E\subseteq \Omega$ with $|E|=m$, and let
\begin{equation}S=\{\alpha\in \mathbf{H}\mid E\nsubseteq\supp(\alpha),\supp(\alpha)\nsubseteq E\}.\end{equation}
$S$ is in fact induced by a hierarchical poset with two levels. More precisely, the poset $\mathbf{P}=(\Omega,\preccurlyeq_{\mathbf{P}})$ defined as
\begin{equation}\text{$i\preccurlyeq_{\mathbf{P}}j$ for all $(i,j)\in E\times(\Omega-E)$}\end{equation}
is hierarchical such that $\min_{\mathbf{P}}(I)=E$ and $\max_{\mathbf{P}}(I)=\Omega-E$; moreover, one can check that
\begin{equation}S=\{\alpha\in \mathbf{H}\mid \supp(\alpha)\not\in\mathcal{I}(\mathbf{P})\}.\end{equation}
When $\mathbb{F}$ is the binary field and $r=1$, necessary and sufficient conditions for the minimality of $\mathcal{C}(S)$ have been given in [53, Theorem 3.2], and we generalize the counterpart result to arbitrary $\mathbb{F}$ and $r$, as detailed in the following theorem.

\setlength{\parindent}{0em}
\begin{theorem}
Let $r\in\{0,1,\dots,n-1\}$. Then, $S$ is a cutting $r$-blocking set of $\mathbf{H}$ if and only if $m\geqslant r+2$ and $n-m\geqslant r+1$.
\end{theorem}

\setlength{\parindent}{2em}
In fact, we have the following more general result.

\setlength{\parindent}{0em}
\begin{theorem}
Let $(N_t\mid t\in\Lambda)$ be a tuple of subsets of $\Omega-E$ indexed by a set $\Lambda$, and let
$$T=\{\theta\in \mathbf{H}\mid E\subseteq\supp(\theta),\forall~t\in\Lambda:\supp(\theta)-E\nsubseteq N_t\}.$$
Then, for $r\in\{0,1,\dots,n-1\}$, $S\cup T$ is a cutting $r$-blocking set of $\mathbf{H}$ if and only if one of the following two conditions hold:

{\bf{(1)}}\,\,$n-m\geqslant r+1$, and $|N_t|\leqslant n-r-2$ for all $t\in\Lambda$;

{\bf{(2)}}\,\,$n-m=r$, $\Lambda=\emptyset$, and either $m=1$ or $q\neq2$ holds.
\end{theorem}

\setlength{\parindent}{0em}
\begin{remark}
In Theorem 5.2, the set $T$ is also induced by the hierarchical poset $\mathbf{P}$ defined in (5.2). More precisely, we have $\{E\cup N_t\mid t\in\Lambda\}\subseteq\mathcal{I}(\mathbf{P})$ and
\begin{eqnarray*}
\hspace*{-2mm}T=\begin{cases}
\{\theta\in \mathbf{H}\mid \supp(\theta)\in\mathcal{I}(\mathbf{P}),\forall~t\in\Lambda:\supp(\theta)\nsubseteq E\cup N_t\},&\Lambda\neq\emptyset;\\
\{\theta\in \mathbf{H}\mid E\subseteq\supp(\theta)\},&\Lambda=\emptyset.
\end{cases}
\end{eqnarray*}
When $q=|\mathbb{F}|=2$ and $r=1$, the minimality of $\mathcal{C}(S\cup T)$ has been explored in part of [28, Theorems 6.1 and 6.2] for $|\Lambda|\in\{1,2\}$, and in [53, Theorem 3.3] for an arbitrary $\Lambda\neq\emptyset$. Theorem 5.2 further generalizes all the counterpart results.
\end{remark}

\setlength{\parindent}{2em}
Next, we consider $\mathcal{C^{'}}(g)$ where $g$ is the Boolean function induced by $S\cup T$ in Theorem 5.2. The minimality and optimality of $\mathcal{C^{'}}(g)$ have been explored in [28, Theorems 6.3 and 6.4] when $|\Lambda|\in\{1,2\}$, and it has been noted in [28, Section 7] that ``it also should be interesting to investigate the
cases of more than two order ideals in hierarchical posets with two levels or many levels''. The following result enables us to settle the minimality of $\mathcal{C^{'}}(g)$ for an arbitrary family of ideals in hierarchical posets with two levels.

\setlength{\parindent}{0em}
\begin{theorem}
Suppose that $q=2$. Define $T$ as in Theorem 5.2, and let $g:\mathbf{H}\longrightarrow\mathbb{F}$ such that $g^{-1}[\{0\}]=S\cup T\cup\{0\}$. Then, $g$ is a cutting $1$-blocking set of $\mathbf{H}\times\mathbb{F}$ if and only if $2\leqslant m\leqslant n-2$, $\bigcup_{t\in\Lambda}N_t=\Omega-E$, and $m=2\Longrightarrow(\forall~t\in\Lambda:N_t\neq\Omega-E)$.
\end{theorem}

\setlength{\parindent}{2em}
In fact, we have the following more general result.

\setlength{\parindent}{0em}
\begin{theorem}
Suppose that $q=2$. Let $W\subseteq\mathbb{F}^{\Omega-E}$, and let
$$T=\{\alpha\in \mathbf{H}\mid E\subseteq\supp(\alpha),(\alpha_i\mid i\in\Omega-E)\in W\}.$$
Moreover, let $g:\mathbf{H}\longrightarrow\mathbb{F}$ such that $g^{-1}[\{0\}]=S\cup T\cup\{0\}$. Then, $g$ is a cutting $1$-blocking set of $\mathbf{H}\times\mathbb{F}$ if and only if one of the following two conditions holds:

{\bf{(1)}}\,\,$3\leqslant m\leqslant n-2$. Moreover, for any $B\leqslant \mathbb{F}^{\Omega-E}$ with $\dim_{\mathbb{F}}(B)=n-m-1$ and any $x\in \mathbb{F}^{\Omega-E}-B$, it holds that $\{x\}+B\nsubseteq W$;

{\bf{(2)}}\,\,$m=2$, $n\geqslant 4$, $W\neq\emptyset$. Moreover, for any $B\leqslant \mathbb{F}^{\Omega-E}$ with $\dim_{\mathbb{F}}(B)=n-3$ and any $x\in \mathbb{F}^{\Omega-E}$, it holds that $\{x\}+B\nsubseteq W$.
\end{theorem}

\setlength{\parindent}{0em}
\begin{remark}
Our approach towards the minimality of $\mathcal{C^{'}}(g)$ is different from those of \cite{28} in the sense that in \cite{28}, the minimality of $\mathcal{C^{'}}(g)$ is explored based on (1.3) by analyzing the Hamming weights of codewords, and our approach is based on Corollary 3.5 and uses cutting blocking sets .
\end{remark}

\setlength{\parindent}{2em}
We first prove Theorem 5.2, and we begin with the following lemma.

\setlength{\parindent}{0em}
\begin{lemma}
Let $V\leqslant_{\mathbb{F}}\mathbf{H}$. Then, there exists $A\leqslant_{\mathbb{F}}\delta(E)$, $B\leqslant_{\mathbb{F}}\delta(\Omega-E)$, $I\leqslant_{\mathbb{F}}\delta(E)$ and $\tau\in\Hom_{\mathbb{F}}(I,\delta(\Omega-E))$ such that $A\cap I=\{0\}$, $\tau$ is injective, $\tau[I]\cap B=\{0\}$ and
$$V=A+B+\{\lambda+\tau(\lambda)\mid\lambda\in I\}.$$
Moreover, if $\langle S\cap V\rangle_{\mathbb{F}}\neq V$, then one of the following three conditions holds:
\begin{equation}B\neq\{0\},~I=\{0\},~\dim_{\mathbb{F}}(A)=1,~\chi(A)=E,\end{equation}
\begin{equation}B\neq\{0\},~\dim_{\mathbb{F}}(I)=1,~A=\{0\},~\chi(I)=E,\end{equation}
\begin{equation}B=\{0\},~\dim_{\mathbb{F}}(I)\leqslant1.\end{equation}
\end{lemma}

\begin{proof}
Let $\pi_1\in\Hom_{\mathbb{F}}(\mathbf{H},\delta(E))$, $\pi_2\in\Hom_{\mathbb{F}}(\mathbf{H},\delta(\Omega-E))$ such that $\alpha=\pi_1(\alpha)+\pi_2(\alpha)$ for all $\alpha\in\mathbf{H}$. Let $A=V\cap\delta(E)$, $B=V\cap\delta(\Omega-E)$, $C=\pi_1[V]$, $D=\pi_2[V]$. Then, there exists an $\mathbb{F}$-isomorphism $\varphi:C/A\longrightarrow D/B$ such that $V=\{\gamma+\theta\mid\gamma\in C,\theta\in D,\varphi(\gamma+A)=\theta+B\}$. Now let $I\leqslant_{\mathbb{F}}C$, $J\leqslant_{\mathbb{F}}D$ such that $C=A+I$, $D=B+J$, $A\cap I=B\cap J=\{0\}$, and define $\tau:I\longrightarrow J$ as $\varphi(\lambda+A)=\tau(\lambda)+B$ for all $\lambda\in I$. One can then check that $\tau:I\longrightarrow J$ is an $\mathbb{F}$-isomorphism satisfying $V=A+B+\{\lambda+\tau(\lambda)\mid\lambda\in I\}$, as desired.

\hspace*{4mm}\,\,From now on, we assume that $\langle S\cap V\rangle_{\mathbb{F}}\neq V$. Noting that $\delta(\Omega-E)-\{0\}\subseteq S$, we have $B-\{0\}\subseteq S\cap V$ and $B\subseteq\langle S\cap V\rangle_{\mathbb{F}}$. Now we prove the statement in the following three steps.

\hspace*{4mm}\,\,First, we show that $\dim_{\mathbb{F}}(I)\leqslant1$. Suppose by way of contradiction that $\dim_{\mathbb{F}}(I)\geqslant2$. Then, we have $I=\langle\{\lambda\in I-\{0\}\mid\supp(\lambda)\subsetneqq E\}\rangle_{\mathbb{F}}$. For any $\lambda\in I-\{0\}$ with $\supp(\lambda)\subsetneqq E$, noticing that $\tau(\lambda)\in\delta(\Omega-E)-\{0\}$, we have $\lambda+\tau(\lambda)\in S\cap V$. It then follows that $\{\lambda+\tau(\lambda)\mid\lambda\in I\}\subseteq\langle S\cap V\rangle_{\mathbb{F}}$. Next, consider $\alpha\in A$. Since $\dim_{\mathbb{F}}(I)\geqslant2$, we can choose $\lambda\in I-\{0\}$ such that $\supp(\alpha+\lambda)\subsetneqq E$. Noticing that $\tau(\lambda)\in\delta(\Omega-E)-\{0\}$, we have $\alpha+\lambda+\tau(\lambda)\in S\cap V$, which further implies that $\alpha\in\langle S\cap V\rangle_{\mathbb{F}}$. The above discussion implies that $\langle S\cap V\rangle_{\mathbb{F}}=V$, a contradiction, as desired. Therefore if $B=\{0\}$, then (5.6) holds. Hence in what follows, we assume that $B\neq\{0\}$.

\hspace*{4mm}\,\,Second, we show that if $A\neq\{0\}$, then it holds that $\dim_{\mathbb{F}}(A)=1$, $\chi(A)=E$ and $A\nsubseteq\langle S\cap V\rangle_{\mathbb{F}}$. Suppose by way of contradiction that $A\subseteq\langle S\cap V\rangle_{\mathbb{F}}$. Consider an arbitrary $\lambda\in I-\{0\}$. Since $A\neq\{0\}$, we can choose $\alpha\in A$ such that $\supp(\alpha+\lambda)\subsetneqq E$. Noticing that $\tau(\lambda)\in\delta(\Omega-E)-\{0\}$, we have $\alpha+\lambda+\tau(\lambda)\in S\cap V$, which, together with $\alpha\in\langle S\cap V\rangle_{\mathbb{F}}$, implies that $\lambda+\tau(\lambda)\in\langle S\cap V\rangle_{\mathbb{F}}$. It then follows that $\langle S\cap V\rangle_{\mathbb{F}}=V$, a contradiction. Therefore we have $A\nsubseteq\langle S\cap V\rangle_{\mathbb{F}}$, as desired. Now suppose by way of contradiction that either $\dim_{\mathbb{F}}(A)\geqslant2$ or $\chi(A)\neq E$ holds. Then, we have $A=\langle\{\alpha\in A\mid\supp(\alpha)\subsetneqq E\}\rangle_{\mathbb{F}}$. Since $B\neq\{0\}$, we can choose $\beta\in B-\{0\}$. For any $\alpha\in A$ with $\supp(\alpha)\subsetneqq E$, it follows from $\alpha+\beta\in S\cap V$ and $\beta\in S\cap V$ that $\alpha\in\langle S\cap V\rangle_{\mathbb{F}}$. It then follows that $A\subseteq\langle S\cap V\rangle_{\mathbb{F}}$, a contradiction, as desired. We then conclude that if $I=\{0\}$, then (5.4) holds.

\hspace*{4mm}\,\,Third, we assume that $I\neq\{0\}$ and establish (5.5). By $\dim_{\mathbb{F}}(I)\leqslant1$, we have $\dim_{\mathbb{F}}(I)=1$. Suppose by way of contradiction that $\chi(I)\subsetneqq E$. For any $\lambda\in I-\{0\}$, from $\supp(\lambda)\subsetneqq E$ and $\tau(\lambda)\in\delta(\Omega-E)-\{0\}$, we deduce that $\lambda+\tau(\lambda)\in S\cap V$. This, together with $B\subseteq\langle S\cap V\rangle_{\mathbb{F}}$ and $\langle S\cap V\rangle_{\mathbb{F}}\neq V$, implies that $A\nsubseteq\langle S\cap V\rangle_{\mathbb{F}}$. Therefore we have $\dim_{\mathbb{F}}(A)=1$, $\chi(A)=E$ and $A\cap\langle S\cap V\rangle_{\mathbb{F}}=\{0\}$. Let $\alpha\in A-\{0\}$. Then, we have $\supp(\alpha)=E$. Since $I\neq\{0\}$, we can choose $\lambda\in I-\{0\}$ such that $\supp(\alpha+\lambda)\subsetneqq E$. Noticing that $\tau(\lambda)\in\delta(\Omega-E)-\{0\}$, we have $\alpha+\lambda+\tau(\lambda)\in S\cap V$, which, together with $\lambda+\tau(\lambda)\in S\cap V$, implies that $\alpha\in\langle S\cap V\rangle_{\mathbb{F}}$, a contradiction, as desired. Hence we have $\chi(I)=E$, and it remains to show that $A=\{0\}$. Suppose by way of contradiction that $A\neq\{0\}$. Then, we have $\dim_{\mathbb{F}}(A)=1$, $\chi(A)=E$ and $A\cap\langle S\cap V\rangle_{\mathbb{F}}=\{0\}$. Consider the $2$-dimensional $\mathbb{F}$-subspace $A+I$. We can choose $\eta_1,\eta_2\in A+I$ such that $\supp(\eta_1)\subsetneqq E$, $\supp(\eta_2)\subsetneqq E$ and $A+I=\langle\{\eta_1,\eta_2\}\rangle_{\mathbb{F}}$. We note that $\eta_1\not\in A$, $\eta_2\not\in A$. Let $\lambda\in I-\{0\}$. Then, there exists $a,b\in\mathbb{F}$ such that $a\eta_1-\lambda\in A$, $b\eta_2-\lambda\in A$. Therefore we have $a\eta_1-b\eta_2\in A$, $a\eta_1+\tau(\lambda)\in V$, $b\eta_2+\tau(\lambda)\in V$. Since $\supp(a\eta_1)\subsetneqq E$, $\supp(b\eta_2)\subsetneqq E$, $\tau(\lambda)\in\delta(\Omega-E)-\{0\}$, we have $a\eta_1+\tau(\lambda)\in S$, $b\eta_2+\tau(\lambda)\in S$, which yields that $a\eta_1-b\eta_2\in\langle S\cap V\rangle_{\mathbb{F}}$. Noticing that $a,b\neq0$ as $\lambda\not\in A$, we have $0\neq a\eta_1-b\eta_2\in A\cap\langle S\cap V\rangle_{\mathbb{F}}$, a contradiction, as desired.
\end{proof}

\setlength{\parindent}{2em}
Now we are ready to prove Theorem 5.2.

\begin{proof}[Proof of Theorem 5.2]First, suppose that (1) holds. Let $V\leqslant_{\mathbb{F}}\mathbf{H}$ such that $\dim_{\mathbb{F}}(V)=n-r$, and write $V=A+B+\{\lambda+\tau(\lambda)\mid\lambda\in I\}$, where $A,B,I,\tau$ are defined as in Lemma 5.1. By $r+1\leqslant n-m$, we have $\dim_{\mathbb{F}}(B)\geqslant n-r-m\geqslant1$, which implies that $B\neq\{0\}$. By Lemma 5.1, we can assume that either (5.4) or (5.5) holds. Therefore we have $\dim_{\mathbb{F}}(B)=n-r-1$. Now we discuss in two cases. First, suppose that (5.4) holds. Let $\alpha\in A-\{0\}$. Then, we have $\supp(\alpha)=E$. Since $\dim_{\mathbb{F}}(B)=n-r-1$, there exists $\beta\in B$ such that $|\supp(\beta)|\geqslant n-r-1$. We note that $E\subseteq\supp(\alpha+\beta)$; moreover, for any $t\in\Lambda$, it follows from $|N_t|\leqslant n-r-2$ that $\supp(\alpha+\beta)-E=\supp(\beta)\nsubseteq N_t$. It follows that $\alpha+\beta\in T\cap V$, which, together with $\beta\in\langle S\cap V\rangle_{\mathbb{F}}$, implies that $\alpha\in\langle (S\cup T)\cap V\rangle_{\mathbb{F}}$. This, together with $B\subseteq\langle S\cap V\rangle_{\mathbb{F}}$ and $I=\{0\}$, implies that $V\subseteq\langle (S\cup T)\cap V\rangle_{\mathbb{F}}$, as desired. Second, suppose that (5.5) holds. Let $\lambda\in I-\{0\}$. Then, we have $\supp(\lambda)=E$. Since $\dim_{\mathbb{F}}(B)=n-r-1$, there exists $\beta\in B$ such that $|\supp(\tau(\lambda)+\beta)|\geqslant n-r-1$. We note that $E\subseteq\supp(\lambda+\tau(\lambda)+\beta)$; moreover, for any $t\in\Lambda$, it follows from $|N_t|\leqslant n-r-2$ that $\supp(\lambda+\tau(\lambda)+\beta)-E=\supp(\tau(\lambda)+\beta)\nsubseteq N_t$. It follows that $\lambda+\tau(\lambda)+\beta\in T\cap V$, which, together with $\beta\in\langle S\cap V\rangle_{\mathbb{F}}$, implies that $\lambda+\tau(\lambda)\in\langle (S\cup T)\cap V\rangle_{\mathbb{F}}$. It then follows from $B\subseteq\langle S\cap V\rangle_{\mathbb{F}}$ and $A=\{0\}$ that $V\subseteq\langle (S\cup T)\cap V\rangle_{\mathbb{F}}$, as desired.

Next, suppose that (2) holds. By $\Lambda=\emptyset$, we have $T=\{\alpha\in\mathbf{H}\mid E\subseteq\supp(\alpha)\}$ and $\mathbf{H}-\delta(E)\subseteq S\cup T$. Let $V\leqslant_{\mathbb{F}}\mathbf{H}$ with $\dim_{\mathbb{F}}(V)=n-r$. If $V\nsubseteq\delta(E)$, then we have $V=\langle V-\delta(E)\rangle_{\mathbb{F}}\subseteq\langle (S\cup T)\cap V\rangle_{\mathbb{F}}$, as desired. If $V\subseteq\delta(E)$, then it follows from $n-r=m$ that $V=\delta(E)$, which implies that $(S\cup T)\cap V=\{\alpha\in\mathbf{H}\mid \supp(\alpha)=E\}$. It then follows from either $|E|=1$ or $|\mathbb{F}|\neq2$ that $V=\langle (S\cup T)\cap V\rangle_{\mathbb{F}}$, as desired.

Now we prove the ``only if'' part. We begin by noting that $n-m\geqslant r$. Indeed, if $n-m\leqslant r-1$, then there exists $A\leqslant_{\mathbb{F}}\delta(E)$ such that $\dim_{\mathbb{F}}(A)=n-r$, $\chi(A)\subsetneqq E$; moreover, it is straightforward to verify that $A\cap(S\cup T)=\emptyset$, and hence $\langle A\cap(S\cup T)\rangle_{\mathbb{F}}=\{0\}\neq A$, a contradiction, as desired.

Next, we show that $|N_t|\leqslant n-r-2$ for all $t\in\Lambda$. Suppose by way of contradiction that $|N_j|\geqslant n-r-1$ for some $j\in\Lambda$. Then, we can choose $B\leqslant_{\mathbb{F}}\delta(N_j)$ such that $\dim_{\mathbb{F}}(B)=n-r-1$; moreover, we can choose $A\leqslant_{\mathbb{F}}\delta(E)$ such that $\dim_{\mathbb{F}}(A)=1$, $\chi(A)=E$. Then, $V\triangleq A+B$ is an $(n-r)$-dimensional $\mathbb{F}$-subspace of $\mathbf{H}$. For $\alpha\in A-\{0\}$ and $\beta\in B$, it follows from $\supp(\alpha)=E$ that $E\subseteq\supp(\alpha+\beta)$ and $\supp(\alpha+\beta)-E=\supp(\beta)\subseteq N_j$, which implies that $\alpha+\beta\not\in S\cup T$. Therefore we have $V\cap(S\cup T)\subseteq B$ and hence $\langle V\cap(S\cup T)\rangle_{\mathbb{F}}\subseteq B\subsetneqq V$, a contradiction, as desired.

Finally, suppose that $n-m=r$. Then, by $\dim_{\mathbb{F}}(\delta(E))=m=n-r$, we have $\delta(E)=\langle (S\cup T)\cap\delta(E)\rangle_{\mathbb{F}}=\langle T\cap\delta(E)\rangle_{\mathbb{F}}$, which implies that $T\cap\delta(E)\neq\emptyset$, which further implies that $\Lambda=\emptyset$, as desired. Now one can check that $T\cap\delta(E)=\{\alpha\in\mathbf{H}\mid\supp(\alpha)=E\}$, which, together with $\delta(E)=\langle T\cap\delta(E)\rangle_{\mathbb{F}}$, further implies that either $m=|E|=1$ or $q=|\mathbb{F}|\neq2$, as desired.
\end{proof}

\setlength{\parindent}{2em}
Next, we derive Theorem 5.1 from Theorem 5.2.

\begin{proof}[Proof of Theorem 5.1]
In Theorem 5.2, set $\Lambda=\{t\}$ and $N_t=\Omega-E$. Then, we have $T=\emptyset$, and hence the desired result follows from the fact that $|N_t|\leqslant n-r-2\Longleftrightarrow m\geqslant r+2$.
\end{proof}

\setlength{\parindent}{2em}
Now we proceed to prove Theorem 5.4. We begin with the following general necessary and sufficient condition for any Boolean function $g:\mathbf{H}\longrightarrow \mathbb{F}$ to be a cutting $r$-blocking set of $\mathbf{H}\times \mathbb{F}$.

\setlength{\parindent}{0em}
\begin{proposition}
Suppose that $q=2$. Let $g:\mathbf{H}\longrightarrow \mathbb{F}$ with $g^{-1}[\{0\}]=L$, and let $r\in\{0,1,\dots,n\}$. Then, $g$ is a cutting $r$-blocking set of $\mathbf{H}\times \mathbb{F}$ if and only if the following three conditions hold:

{\bf{(1)}}\,\,For any $A\leqslant_{\mathbb{F}}\mathbf{H}$ with $\dim_{\mathbb{F}}(A)=n-r$, it holds that $A\nsubseteq L$, and that $L\cap A$ is not an $(n-r-1)$-dimensional $\mathbb{F}$-subspace of $\mathbf{H}$;

{\bf{(2)}}\,\,For any $B\leqslant_{\mathbb{F}}\mathbf{H}$ with $\dim_{\mathbb{F}}(B)=n-r+1$, it holds that $\langle L\cap B\rangle_{\mathbb{F}}=B$;

{\bf{(3)}}\,\,For any $A,B\leqslant_{\mathbb{F}}\mathbf{H}$ such that $\dim_{\mathbb{F}}(A)=n-r$, $\dim_{\mathbb{F}}(B)=n-r+1$, $A\subseteq B$, it holds that $\langle L\cap A\rangle_{\mathbb{F}}+\langle B-(L\cup A)\rangle_{\mathbb{F}}=B$.
\end{proposition}

\begin{proof}
We begin by noting that $\{W\leqslant_{\mathbb{F}}\mathbf{H}\times \mathbb{F}\mid\dim_{\mathbb{F}}(W)=n+1-r\}=\Delta_1\cup\Delta_2\cup\Delta_3$, where $\Delta_1=\{A\times\mathbb{F}\mid A\leqslant_{\mathbb{F}}\mathbf{H},\dim_{\mathbb{F}}(A)=n-r\}$, $\Delta_2=\{B\times\{0\}\mid B\leqslant_{\mathbb{F}}\mathbf{H},\dim_{\mathbb{F}}(B)=n+1-r\}$ and
$$\Delta_3=\{A\times\{0\}\cup(B-A)\times\{1_{\mathbb{F}}\}\mid A,B\leqslant_{\mathbb{F}}\mathbf{H},\dim_{\mathbb{F}}(A)=n-r,\dim_{\mathbb{F}}(B)=n-r+1,A\subseteq B\}.$$
First, for $A\leqslant_{\mathbb{F}}\mathbf{H}$ with $\dim_{\mathbb{F}}(A)=n-r$, one can check that $g\cap(A\times\mathbb{F})=\{(x,g(x))\mid x\in A\}$, which implies that $\langle g\cap(A\times\mathbb{F})\rangle_{\mathbb{F}}=A\times\mathbb{F}$ if and only if $g\mid_{A}\not\in\Hom_{\mathbb{F}}(A,\mathbb{F})$, if and only if $A\nsubseteq L$, and $L\cap A$ is not an $(n-r-1)$-dimensional $\mathbb{F}$-subspace of $\mathbf{H}$, as desired. Second, for $B\leqslant_{\mathbb{F}}\mathbf{H}$ with $\dim_{\mathbb{F}}(B)=n+1-r$, one can check that $\langle g\cap (B\times\{0\})\rangle_{\mathbb{F}}=\langle B\cap L\rangle_{\mathbb{F}}\times\{0\}$, which implies that $\langle g\cap(B\times\{0\})\rangle_{\mathbb{F}}=B\times\{0\}$ if and only if $\langle B\cap L\rangle_{\mathbb{F}}=B$, as desired. Third, for $A,B\leqslant_{\mathbb{F}}\mathbf{H}$ such that $\dim_{\mathbb{F}}(A)=n-r$, $\dim_{\mathbb{F}}(B)=n-r+1$, $A\subseteq B$, and $W\triangleq A\times\{0\}\cup(B-A)\times\{1_{\mathbb{F}}\}$, one can check that $\langle g\cap W\rangle_{\mathbb{F}}=\{(x,W(x))\mid x\in \langle L\cap A \rangle_{\mathbb{F}}+\langle B-(A\cup L)\rangle_{\mathbb{F}}\}$, which implies that $\langle g\cap W\rangle_{\mathbb{F}}=W$ if and only if $\langle L\cap A \rangle_{\mathbb{F}}+\langle B-(A\cup L)\rangle_{\mathbb{F}}=B$, as desired.
\end{proof}

\setlength{\parindent}{2em}
\begin{corollary}
Suppose that $q=2$. Let $g:\mathbf{H}\longrightarrow \mathbb{F}$ with $g^{-1}[\{0\}]=L$, and let $r\in\{0,1,\dots,n-1\}$. If $g$ is a cutting $r$-blocking set of $\mathbf{H}\times \mathbb{F}$, then for any $A\leqslant_{\mathbb{F}}\mathbf{H}$ with $\dim_{\mathbb{F}}(A)=n-r$ and any $x\in \mathbf{H}$, it holds that $\{x\}+A\nsubseteq L$. Conversely, if $L$ is a cutting $r$-blocking set of $\mathbf{H}$ and $\{x\}+A\nsubseteq L$ for any $A\leqslant_{\mathbb{F}}\mathbf{H}$ with $\dim_{\mathbb{F}}(A)=n-r$ and any $x\in \mathbf{H}$, then $g$ is a cutting $r$-blocking set of $\mathbf{H}\times \mathbb{F}$.
\end{corollary}

\begin{proof}
First, suppose that $g$ is a cutting $r$-blocking set of $\mathbf{H}\times \mathbb{F}$. Let $A\leqslant_{\mathbb{F}}\mathbf{H}$ with $\dim_{\mathbb{F}}(A)=n-r$, and let $x\in \mathbf{H}$. If $x\in A$, then it follows from Proposition 5.1 that $\{x\}+A=A\nsubseteq L$. If $x\not\in A$, then $B\triangleq A\cup(\{x\}+A)$ is an $(n-r+1)$-dimensional $\mathbb{F}$-subspace of $\mathbf{H}$ with $A\subseteq B$, which, together with Proposition 5.1, implies that $\langle L\cap A\rangle_{\mathbb{F}}+\langle B-(L\cup A)\rangle_{\mathbb{F}}=B$, which further implies that $B-(L\cup A)\neq\emptyset$, and hence $\{x\}+A=B-A\nsubseteq L$, as desired.

Second, we prove the ``conversely'' part. Indeed, Condition (1) of Proposition 5.1 follows from the assumption; moreover, if $r\geqslant1$, then it follows from Proposition 2.1 that $L$ is a cutting $(r-1)$-blocking set of $\mathbf{H}$, which establishes Condition (2) of Proposition 5.1. Next, let $A,B\leqslant_{\mathbb{F}}\mathbf{H}$ such that $\dim_{\mathbb{F}}(A)=n-r$, $\dim_{\mathbb{F}}(B)=n-r+1$, $A\subseteq B$. Then, we have $\langle L\cap A\rangle_{\mathbb{F}}=A$. Let $x\in B-A$. Then, we have $B-A=\{x\}+A\nsubseteq L$, which implies that $B\nsubseteq L\cup A$. It then follows that $\langle B-(L\cup A)\rangle_{\mathbb{F}}\nsubseteq A$, which, together with $\dim_{\mathbb{F}}(A)=\dim_{\mathbb{F}}(B)-1$, implies that $\langle L\cap A\rangle_{\mathbb{F}}+\langle B-(L\cup A)\rangle_{\mathbb{F}}=B$, which further establishes Condition (3) of Proposition 5.1, as desired.
\end{proof}

\setlength{\parindent}{2em}
Now we are ready to prove Theorem 5.4.

\begin{proof}[Proof of Theorem 5.4]
Throughout the proof, we let $L=S\cup T\cup\{0\}$, $\eta\in\mathbf{H}$ with $\supp(\eta)=E$, and $M=\{\alpha\in\delta(\Omega-E)\mid(\alpha_i\mid i\in\Omega-E)\in W\}$. It can be readily verified that $T=\{\eta\}+M$.

First of all, let $V\leqslant_{\mathbb{F}}\mathbf{H}$ such that $\delta(E)\subseteq V$, and let $J=\delta(\Omega-E)\cap V$; moreover, let $x\in\delta(\Omega-E)-J$. \textbf{We claim that $\{x\}+V\subseteq L\Longleftrightarrow\{x\}+J\subseteq M$}. Indeed, for any $\alpha\in\delta(E)-\{\eta\}$ and $\beta\in J$, by $x+\beta\in\delta(\Omega-E)-\{0\}$ and $\supp(\alpha)\subsetneqq E$, we have $\alpha+x+\beta\in S$; moreover, for any $\theta\in J$, by $E\subseteq \supp(\eta+x+\theta)$, we have $\eta+x+\theta\not\in S\cup\{0\}$. It then follows from $V=\delta(E)+J$ that $\{x\}+V\subseteq L$ if and only if $\{\eta+x\}+J\subseteq T$, if and only if $\{x\}+J\subseteq M$, as desired.

Next, we show that if either (1) or (2) holds, then for any $V\leqslant_{\mathbb{F}}\mathbf{H}$ with $\dim_{\mathbb{F}}(V)=n-1$ and any $x\in \mathbf{H}-V$, it holds that $\{x\}+V\nsubseteq L$. We discuss in two cases. First, assume that $\delta(E)\nsubseteq V$. Then, by $\delta(E)+V=\mathbf{H}$, there exists $y\in\delta(E)-V$ such that $\{x\}+V=\{y\}+V$. Since $\dim_{\mathbb{F}}(\delta(E)\cap V)=m-1$, we can choose $\gamma\in\delta(E)\cap V$ such that $|\supp(y-\gamma)|\leqslant1$. From $m\geqslant2$, we deduce that $\supp(y-\gamma)\subsetneqq E$, which implies that $y-\gamma\not\in L$ and $y-\gamma\in\{x\}+V$, as desired. Second, assume that $\delta(E)\subseteq V$. Then, $J\triangleq\delta(\Omega-E)\cap V$ is an $(n-m-1)$-dimensional $\mathbb{F}$-subspace of $\delta(\Omega-E)$. Since $x\not\in V$, we can choose $y\in\delta(\Omega-E)-J$ such that $\{y\}+V=\{x\}+V$. By either (1) or (2), we have $\{y\}+J\nsubseteq M$, which implies that $\{x\}+V=\{y\}+V\nsubseteq L$, as desired.

Now assume that (1) holds. By Theorem 5.2, $L$ is a cutting $1$-blocking set of $\mathbf{H}$. Let $V\leqslant_{\mathbb{F}}\mathbf{H}$ with $\dim_{\mathbb{F}}(V)=n-1$. Since $\dim_{\mathbb{F}}(\delta(E)\cap V)\geqslant m-1$, there exists $\gamma\in(\delta(E)\cap V)-\{0\}$ such that $|\supp(\gamma)|\leqslant2$. From $m\geqslant3$, we deduce that $\supp(\gamma)\subsetneqq E$, which implies that $\gamma\not\in L$, and hence $V\nsubseteq L$. It then follows from Corollary 5.1 that $g$ is a cutting $1$-blocking set of $\mathbf{H}\times\mathbb{F}$, as desired.

Next, assume that (2) holds. We claim that $L$ is a cutting $1$-blocking set of $\mathbf{H}$. Indeed, let $V\leqslant_{\mathbb{F}}\mathbf{H}$ with $\dim_{\mathbb{F}}(V)=n-1$. If $V\neq\{0,\eta\}+\delta(\Omega-E)$, then it follows from Lemma 5.1 that $\langle S\cap V\rangle_{\mathbb{F}}=V$, as desired. If $V=\{0,\eta\}+\delta(\Omega-E)$, then by $W\neq\emptyset$, we can choose $\theta\in M$; moreover, we infer that $\eta+\theta\in T\cap V$, which, together with $\delta(\Omega-E)\subseteq (S\cup\{0\})\cap V$, implies that $\langle L\cap V\rangle_{\mathbb{F}}=V$, as desired. Now we discuss in two cases to show that $V\nsubseteq L$. First, suppose that $\eta\not\in V$. Noticing that $\delta(E)\cap V\neq\{0\}$, we infer that $\emptyset\neq(\delta(E)\cap V)-\{0\}\subseteq V-L$, as desired. Second, suppose that $\eta\in V$. Noticing that $\dim_{\mathbb{F}}(\delta(\Omega-E)\cap V)\geqslant n-3$, by (2), we have $\delta(\Omega-E)\cap V\nsubseteq M$. Hence we can choose $\beta\in (\delta(\Omega-E)\cap V)-M$. We note that $\eta+\beta\not\in T$; moreover, it follows from $E\subseteq\supp(\eta+\beta)$ that $\eta+\beta\not\in S\cup\{0\}$, which, implies that $\eta+\beta\in V-L$, as desired. It then follows from Corollary 5.1 that $g$ is a cutting $1$-blocking set of $\mathbf{H}\times\mathbb{F}$.

By now, we have established the ``if'' part, and we next proceed to establish the ``only if'' part.

First, we show that $2\leqslant m\leqslant n-2$. Indeed, if $m=1$, then we have $S\cup\{0\}=\delta(\Omega-E)$ and $\dim_{\mathbb{F}}(\delta(\Omega-E))=n-1$, a contradiction to (1) of Proposition 5.1, as desired. Now suppose by way of contradiction that $2\leqslant m=n-1$. Let $\theta\in\mathbf{H}$ with $\supp(\theta)=\Omega$. Then, we have $\dim_{\mathbb{F}}(\delta(E))=n-1$, $L\cap\delta(E)\subseteq\{0,\eta\}$ and $\mathbf{H}-(L\cup\delta(E))\subseteq\{\eta,\theta\}$, which further implies that $\langle L\cap\delta(E)\rangle_{\mathbb{F}}+\langle\mathbf{H}-(L\cup\delta(E))\rangle_{\mathbb{F}}\subsetneqq\mathbf{H}$, a contradiction to (3) of Proposition 5.1, as desired.

Next, let $J\leqslant\delta(\Omega-E)$ with $\dim_{\mathbb{F}}(J)=n-m-1$. For $x\in \delta(\Omega-E)-J$ and $V\triangleq\delta(E)+J$, since $\dim_{\mathbb{F}}(V)=n-1$, $\delta(E)\subseteq V$ and $J=\delta(\Omega-E)\cap V$, by Corollary 5.1, we have $\{x\}+V\nsubseteq L$, which implies that $\{x\}+J\nsubseteq M$. Therefore if $3\leqslant m\leqslant n-2$, then (1) holds true. Hence from now on, we suppose that $m=2$. It suffices to show that $W\neq\emptyset$ and $J\nsubseteq M$. Indeed, $A\triangleq\{0,\eta\}+\delta(\Omega-E)$ satisfies that $\dim_{\mathbb{F}}(A)=n-1$ and $A\cap(S\cup\{0\})=\delta(\Omega-E)$; moreover, by $\dim_{\mathbb{F}}(\delta(\Omega-E))=n-2$ and Proposition 5.1, we have $A\cap L\neq\delta(\Omega-E)$. It then follows that $T\neq\emptyset$, and hence $W\neq\emptyset$, as desired. Now choose $\gamma\in\delta(E)-\{0,\eta\}$, $\theta\in\delta(\Omega-E)-J$, and let $X=\{0,\eta\}+J+\{0,\gamma+\theta\}$. For $\beta\in J$, we have $\beta\in S\cup\{0\}$; moreover, it follows from $\supp(\gamma)\subsetneqq E$, $\supp(\eta+\gamma)\subsetneqq E$ and $\theta+\beta\in\delta(\Omega-E)-J$ that $\gamma+\theta+\beta\in S$, $\eta+\gamma+\theta+\beta\in S$. However, by $\dim_{\mathbb{F}}(X)=n-1$ and Corollary 5.1, we have $X\nsubseteq L$, which implies that $\{\eta\}+J\nsubseteq T$, and hence $J\nsubseteq M$, as desired.
\end{proof}

\setlength{\parindent}{2em}
Finally, we derive Theorem 5.3 from Theorem 5.4.

\begin{proof}[Proof of Theorem 5.3]
We begin by noting that $T=\{\alpha\in \mathbf{H}\mid E\subseteq\supp(\alpha),(\alpha_i\mid i\in\Omega-E)\in W\}$, where $W=\{\gamma\in\mathbb{F}^{\Omega-E}\mid\forall~t\in\Lambda:\supp(\gamma)\nsubseteq N_t\}$, and $\supp(\gamma)=\{i\in\Omega-E\mid\gamma_i\neq0\}$
for all $\gamma\in\mathbb{F}^{\Omega-E}$.

Next, we prove the ``only if'' part. By Theorem 5.4, we have $2\leqslant m\leqslant n-2$; moreover, if $m=2$, then we have $W\neq\emptyset$, which implies that $(\forall~t\in\Lambda:N_t\neq\Omega-E)$, as desired. Now consider $i\in\Omega-E$. Let $\gamma\in\mathbb{F}^{\Omega-E}$ with $\supp(\gamma)=\{i\}$, and let $B=\{\rho\in\mathbb{F}^{\Omega-E}\mid\rho_i=0\}$. By Theorem 5.4, we have $\{\gamma\}+B\nsubseteq W$. Hence we can choose $\rho\in B$ such that $\gamma+\rho\not\in W$. Therefore $\supp(\gamma+\rho)\subseteq N_t$ for some $t\in\Lambda$. From $i\in\supp(\gamma+\rho)$, we deduce that $i\in \bigcup_{t\in\Lambda}N_t$, as desired.

Now we prove the ``if'' part. Let $B\leqslant\mathbb{F}^{\Omega-E}$ with $\dim_{\mathbb{F}}(B)=n-m-1$, and let $\gamma\in\mathbb{F}^{\Omega-E}$. Then, we can choose $\beta\in B$ such that $|\supp(\gamma-\beta)|\leqslant|\Omega-E|-\dim_{\mathbb{F}}(B)=1$. Since $\Omega-E=\bigcup_{t\in\Lambda}N_t$, there exists $t\in\Lambda$ such that $\supp(\gamma-\beta)\subseteq N_t$, which implies that $\gamma-\beta\not\in W$, and hence $\{\gamma\}+B\nsubseteq W$. Moreover, if $m=2$, then we have $(\forall~t\in\Lambda:N_t\neq\Omega-E)$, which implies that $W\neq\emptyset$. It then follows from Theorem 5.4 that $g$ is a cutting $1$-blocking set of $\mathbf{H}\times\mathbb{F}$, as desired.
\end{proof}

\end{document}